\documentclass[aps,amssymb,amsmath,amsfonts,superscriptaddress]{revtex4}

\usepackage[bookmarksnumbered, colorlinks,plainpages]{hyperref}
\hypersetup{colorlinks,linkcolor=blue,citecolor=blue,urlcolor=blue}
\usepackage{epsf,epsfig}
\DeclareMathAlphabet\mathzapf{T1}{pzc}{mb}{it}
\usepackage{MnSymbol}
\usepackage[caption=false]{subfig}
\usepackage{color,soul}
\usepackage{xcolor}
\usepackage{pgf,tikz}
\usepackage{graphicx}
\usepackage{dcolumn}
\usepackage{bm,bbm}
\usepackage{color,soul}
\usepackage{xcolor}
\usepackage{import}
\usepackage{csquotes}
\usepackage{amsthm}
\usepackage[toc,page]{appendix}
\graphicspath{{./fig/}}
\usepackage{enumerate}
\usepackage{listings}
\usepackage{soul}
\usepackage[]{xcolor}
\usepackage{chngcntr}

\newtheorem{theorem}{Theorem}
\newtheorem{lemma}[theorem]{Lemma}
\newtheorem{proposition}[theorem]{Proposition}
\newtheorem{corollary}[theorem]{Corollary}
\newtheorem{remark}[theorem]{Remark}
\newtheorem{example}{Example}

\newcommand{\Tr}{{\mathrm{Tr}}}

\newcommand{\1}{{\rm 1\hspace{-0.9mm}l}}

\DeclareMathOperator{\diag}{diag}

\newtheorem{definition}[theorem]{Definition}

\newcommand{\ket}[1]{|#1\rangle}
\newcommand{\bra}[1]{\langle #1|}

\newcommand{\project}[1]{\ket{#1}\bra{#1}}

\newcommand{\outprod}[2]{\ket{#1}\bra{#2}}

\newcommand{\e}{{\mathrm e}}
\renewcommand{\c}[1]{\mathcal{#1}}
\renewcommand{\sf}[1]{\mathsf{#1}}

\renewcommand{\exp}{{\mathrm {exp}}}

\newcommand{\vect}[1]{\mathbf{#1}}

\newcommand{\tr}[1]{\mbox{Tr} #1}

\newcommand{\bl}[1]{{\color{blue}#1}}
\newcommand{\violet}[1]{{\color{violet}#1}}

\begin{document}

\title{Accessible maps in a group of classical or quantum channels} 

\author{Koorosh Sadri}
\affiliation{Department of Physics, Sharif University of Technology, Tehran, Iran}
\author{Fereshte Shahbeigi}
\affiliation{Department of Physics, Sharif University of Technology, Tehran, Iran}
\affiliation{Center for Theoretical Physics, Polish Academy of Sciences, 02-668 Warsaw, Poland}
\author{Zbigniew Pucha{\l}a}
\affiliation{Institute of Theoretical and Applied Informatics, Polish Academy
of Sciences, 44-100 Gliwice, Poland}
\affiliation{Faculty of Physics, Astronomy and Applied Computer Science,
Jagiellonian University,  30-348 Krak{\'o}w, Poland}
\author{Karol {\.Z}yczkowski}
\affiliation{Faculty of Physics, Astronomy and Applied Computer Science,
Jagiellonian University, 30-348 Krak{\'o}w, Poland}
  \affiliation{Center for Theoretical Physics, Polish Academy of Sciences, 02-668 Warsaw, Poland}

\date{January 28, 2022}

\begin{abstract}
\noindent 
We study the problem of accessibility in a set of classical and quantum channels
admitting a group structure. Group properties of the set of channels, and the
structure of the closure of the analyzed group $G$ plays a pivotal role in this
regard. The set of all convex combinations of the group elements contains a
subset of channels that are accessible by a dynamical semigroup. We demonstrate
that accessible channels are determined by probability vectors of weights of a
convex combination of the group elements, which depend neither on the dimension
of the space on which the channels act, nor on the specific representation of
the group. Investigating geometric properties of the set $\cal A$  of accessible
maps we show that this set is non-convex, but it enjoys the star-shape property
with respect to the uniform mixture of all elements of the group. We demonstrate
that the set $\cal A$ covers a positive volume in the polytope of all convex
combinations of the elements of the group.
\end{abstract}

\maketitle

{\sl Dedicated to the memory of  Prof. Andrzej Kossakowski (1938 -- 2021)}

\section{Introduction}

The Schr\"odinger equation describes the time evolution of an isolated quantum
system. However, for a system open to an environment, one needs a broader
framework. Gorini, Kossakowski and Sudarshan~\cite{GKS76}, and independently
Lindblad~\cite{Li76}, derived such an equation of motion governing the evolution
of open quantum systems -- see \cite{CP17} for historical remarks. The most general
form of such a GKLS generator, acting on $N$ dimensional systems in Hilbert 
space $\c H_N$ and implying Markovian dynamics, is given by
\begin{equation}\label{gkls}
L(\cdot)=-i[H,\cdot]+\Lambda(\cdot)-\frac12\{\Lambda^*(\1),\cdot\},
\end{equation}

where $H$ is the effective Hamiltonian, $\Lambda$ is a completely positive map,
and $*$ denotes the dual in the Heisenberg picture. The first term on the
right-hand side is responsible for the unitary part of the evolution, whereas
the remaining terms describe the dissipation. Any operation $L$ may be written
in the GKLS form \eqref{gkls} if and only if it is (a) Hermiticity preserving,
$L(X^\dagger)=L(X)^\dagger$, (b) trace suppressing, $L^*(\1)=0$, and (c)
conditionally completely positive, i.e. $L\otimes\1(\project{\psi_+})$  is
positive semidefinite in the subspace orthogonal to the maximally entangled
state, $\ket{\psi_+}=\sum\ket{ii}/\sqrt{N}$ \cite{WECC08}. All terms in the
above equation can generally be time-dependent. While a time-dependent Lindblad
operation, $L_t$, governs a broader class of evolutions, the simple structure of
a semigroup is predicated by a time-independent case, $L_t=L$.\\

On the other hand, one may adopt yet another approach, namely quantum channels,
applicable in a more general situation. Indeed, this method works whenever the
physical system and its interacting the environment are initially disjoint
\cite{RH11}, and even correlated in some cases
\cite{SL09,SL16,Betal13,B14,DSL16}. By a quantum channel, we mean a completely
positive and trace-preserving map sending the convex set of $N$-dimensional
quantum states into itself. Assuming that the time evolution of the system and 
environment is governed by the time-dependent unitary operator $U_t$, one has
$\c E_t(\rho)=\Tr_{E}\left[U_t(\rho\otimes\sigma)U^\dagger_t\right]$. The
Stinespring dilation \cite{S55} then guarantees that $\c E_t$ is a quantum
channel. Accordingly, quantum channels may cover a broader set of evolutions
rather than Markovian ones. Note that since $U_t$ satisfies the Schr\"odinger
equation, $\c E_t$ is a one parameter quantum channel continuous on $t$.\\

The question of Markovianity is then brought up. It was introduced into quantum
theory in 2008 \cite{WECC08}; however, the corresponding classical problem has a
long history \cite{K62,Ru62,C95,D10}. The question,  if a given classical map
described by a stochastic matrix is {\sl embeddable}, so a continuous Markov
process can generate it remains open, and it was recently generalized for the
quantum case \cite{KL21}.\\

Markovianity asks for a given quantum channel $\c E$ whether there exists a
Lindblad generator of the form \eqref{gkls} whose resulting quantum map is equal
to $\c E$ at some time $t$. For a general time-dependent Lindblad generator, the
answer to this question is positive if and only if  $\c E$ is an infinitesimal
divisible channel \cite{D89,WC08}. By definition, an infinitesimal divisible
channel is one that can be written as a concatenation of quantum maps
arbitrarily close to identity \cite{D89}.\\

Here, we are more interested in the case with a time-independent Lindblad
generator. Thus the central question for an assumed quantum channel $\c E$ is to
find whether there exists a dynamical semigroup starting from identity and equal
to $\c E$ at some time $t=T$, i.e. if $\c E=\exp\{TL\}$ where $L$ is a
time-independent Lindblad generator in the form \eqref{gkls}. Equivalently, one
may ask if there exists a logarithm for a quantum channel $\c E$ satisfying
properties (a)-(c) mentioned above 
\cite{WECC08}. However, the non-uniqueness of logarithm for an assumed matrix
leaves this a highly difficult question. Indeed, it has been proved to be an
NP-hard problem from the computational perspective as well \cite{CEW12}.\\

As a consequence, while for qubit channels more facts are revealed
\cite{DZP18,PRZ19,CC19,JSP20}, less is known about Markovianity of quantum
channels in higher dimensions \cite{WECC08,SCh17,SAPZ20,Si21}. For instance, in
the case of Pauli channels, $\c E(\rho)=\sum_{i=0}^3
p_i\sigma_i\rho\sigma_i^\dagger$ where $\sigma_0=\1_2$ and $\sigma_i$ for
$i\in\{1,2,3\}$ is one of Pauli matrices, the subset of Markovian channels can
be divided into two classes \cite{DZP18}. The first class contains three
measure-zero subsets of the channels with two negative degenerated eigenvalues
$\lambda_-$ satisfying $\lambda_-^2\leq\lambda_+$ where $\lambda_+$ is the other
nontrivial positive eigenvalue of the channel. While the second class is a
subset of channels with positive eigenvalues such that eigenvalues satisfy
$\lambda_i\lambda_j\leq\lambda_k$ for all combinations of different $i,j,k$. The
latter set occupies $3/32$ of the whole set of Pauli channels \cite{SAPZ20}.
Interestingly, in the case of Pauli channels by taking a convex combination of
three Lindblad generators of the form $L_i(\rho)=\sigma_i\rho\sigma_i^\dagger
-\rho$ for $i=1,2,3$ as the generator of the semigroup, one can exhaust the
entire volume of Markovian maps \cite{PRZ19}. Motivated by this fact, a slightly
different and  simplified version of the Markovianity problem called
accessibility has been introduced~\cite{SAPZ20}. Given a set of 
quantum channels $S=\{\c E_i\}$ with a not necessarily finite number of
elements, in the accessibility problem one asks which quantum channels can be
generated by a Lindblad generator of the form $L(\rho)=\sum q_i \c
E_i(\rho)-\rho$ for some probabilities $q_i$.\\

The accessibility problem was studied for mixed unitary channels given by a
convex combination of Weyl unitary operators \cite{W27} -- a unitary
generalization of Pauli matrices to higher dimensions -- which form the set $S$
of extremal channels \cite{SAPZ20}. Thus accessible channels form a subset $\cal
A$ of the polytope of all convex combinations of the elements of the generating
set $S$. The set $\cal A$ can be obtained by proper Lindblad generators related
to the elements of the set $S$ of Weyl channels itself. Although the set $\cal
A$ of accessible channels is a proper subset of Markovian ones by definition,
like the Pauli channels, the accessible Weyl channels recover the full measure
of Markovian maps. This volume is reduced by an increase in the dimension of the
system \cite{SAPZ20}. In the general case, the set $\cal A$ of accessible maps
occupies a positive volume of the set of channels.\\

Another relevant fact about Markovian and accessible Pauli channels concerns the
rank of these channels. It is known that Pauli channels of Choi rank $3$ are
neither accessible nor Markovian \cite{DZP18,PRZ19}. In higher dimensions, to
our best knowledge, the rank problem has been solved  only  for a mixture of
qutrit Weyl channels accessible by a Lindblad semigroup confirming the existence
of accessible channels only of rank $1,3,9$ \cite{SAPZ20}. The {\it
accessibility rank} has been an open problem for Weyl channels in other
dimensions as well as other quantum channels. \\

The aim of this work is to study the problem of accessibility for a large class
of quantum channels admitting a semigroup structure. Such a set contains Weyl
channels as a special subset. We show that the group properties play an
essential role in the accessibility problem. Applying such an approach, we can
solve the accessibility rank for Weyl channels of any dimension and show which
channels with admissible accessible ranks are Markovian. Furthermore, we obtain
analytic results for the relative volume of the set $\cal A$ of accessible
channels in some cases. \\


The structure of the paper is as follows. In Section \ref{general} the necessary
notions of Lindblad dynamics and accessible channels are introduced. Subsequent
Section \ref{Examples} provides a detailed description of the set $\cal A$ of
accessible maps generated by a given group $G$ of quantum channels. Some key
results of this work are presented in Section \ref{geometry}, in which we show
that the set $\cal A$ is non-convex, but it has the star-shape property with
respect to the uniform mixture of all elements of the group $G$. The following
section discusses some further results on quantum channels, while the case of
stochastic matrices and classical maps is described in Sec. \ref{classical}. The
final Section \ref{conclusion}  concludes the work and presents a list of open
problems. Derivation of the relative volume of the set $\cal A$  for cyclic and
non-cyclic groups of order $g=4$ is provided in Appendix \ref{AppA}.

\section{General Picture} \label{general}
Consider a quantum channel $\c E$ acting on  $N$ dimensional states and
described by a set $\{K_a\}$ of Kraus operators. The superoperator is then
represented by a matrix of order $N^2$, which reads $\Phi_{\c E} =\sum_a
K_a\otimes\overline{K}_a$, where overline denotes complex conjugation.\\

Let $B_\alpha$ with $\alpha\in\{0,\dots,N^2-1\}$ 
form  a basis of $N^2$ matrices 
of order $N$ such that
$\Tr ( B_\alpha B^\dagger_\beta)=\delta_{\alpha\beta}$
for any  $ \alpha,\beta$.
Hence  any operator $X$ of dimension $N$
can be expanded as $X=\sum_{\alpha} x_\alpha B_\alpha$. 
If a quantum channel $\c E$ acts on  $X$ as the input, 
the output  reads $\c E(X)=\sum (\Phi_{\c E}\bm{x})_\alpha B_\alpha$.
Here the $N^2$ dimensional matrix
$(\Phi_{\c E})_{\alpha,\beta}=
\Tr\left(B^\dagger_\alpha\ \c E(B_\beta)\right)$ 
introduces the channel effects to
the vector $\bm{x}$ of same dimension formed by
 coefficients $x_\alpha$ of the input.
Selecting $B_\alpha=B_{ij}=\outprod ij$ as the basis,
 the superoperator  takes the standard form 
$\sum K_a\otimes\overline{K}_a$.
Using the Hermitian basis of generalized Gell-Mann matrices,
with $B_0=\frac{1}{\sqrt{N}}\1_N$, the affine parameterization of
quantum channels in terms of the generalized Bloch vector is obtained,
\begin{equation}\label{affine}
\Phi_\c E=
\begin{pmatrix}
  1       && 0\\
  \vect{t}&& M
\end{pmatrix}.
\end{equation}
Here $M$ denotes a real matrix of dimension $N^2-1$ called the distortion
matrix, while $\vect{t}$ is a real translation vector of the same dimension
presenting how the channel $\c E$ shifts the identity. The same approach gives
the superoperator assigned  to a Lindblad generator $L$ given in \eqref{gkls},
as $(\c L_L)_{\alpha,\beta}= \Tr\left(B^\dagger_\alpha L(B_\beta)\right)$.
Hereafter, quantum channels and  Lindblad generators are denoted by their
corresponding superoperators, $\Phi$ and $\c L$,  and we will drop the
subscripts $\c E$ and $L$  for the sake of brevity.\\

If $\Phi$ denotes a quantum channel, then $\c L_\Phi=\Phi-\1$ is a Hermiticity
preserving and trace suppressing map. Moreover, complete positivity of $\Phi$
imposes conditional complete positivity on $\c L_\Phi$ resulting in a valid
Lindblad generator assigned to any quantum channel \cite{WC08}. It is possible
to show this fact through Eq.~\eqref{gkls} by choosing $H=\1_N$ and $\Lambda=\c
E$. Being trace preserving, $\c E$ admits $\Lambda^\ast(\1)=\c E^\ast(\1)=\1$
which implies $L_{\c E}(\cdot)=\c E(\cdot)-\frac12\{\1_N,\cdot\}$. However, the 
inverse is not valid, i.e. not all Lindblad generators can be obtained by 
subtracting the identity from a quantum channel.
The operation $\c L_\Phi=\Phi-\1$ is a generator of a dynamical semigroup,
$\e^{t\c L_\Phi}$, which provides a completely positive and trace-preserving map
at each moment, $t\geq0$, starting from $\e^{0\c L_\Phi}=\1$ and tending to a
point in the set of quantum channels at $t\to\infty$. The quantum channel
$\Phi$, based on which $\c L_\Phi$ is defined, does not necessarily belong to
the trajectory of $\e^{t\c L_\Phi}$.

\begin{proposition}\label{limit}
Let $\Phi$ denote a quantum channel and $\c L_\Phi=\Phi-\1$ be the corresponding
Lindblad generator. The quantum channel $\e^{t\c L_\Phi}$ at $t\to\infty$ is a
projective map that preserves the invariant states of $\Phi$.
\end{proposition}
\begin{proof}
According to the Perron--Frobenius theorem,  the spectrum of a superoperator
$\Phi$ corresponding to a completely positive and trace-preserving map is
confined to the unit disk, and it contains a leading eigenvalue equal to unity.
Hence the eigenvalues of the Lindblad generator $\mathcal{L}_\Phi =  \Phi-\1$
are either zero or have negative real parts. Therefore, the exponential
$e^{t\mathcal{L}_\Phi}$ approaches the operator projecting on the vector
subspace on which $\Phi$ acts as identity.
This subspace is spanned by the the eigenvectors $\bm{x}$ satisfying $\Phi\bm{x}
= \bm{x}$. As an example, note that a unitary channel, preserves the diagonal
entries and therefore has (at least) $N$ such eigenvectors; in such cases,
$\lim_{t\rightarrow \infty}e^{t\mathcal{L}_U}$ will therefore be the decoherence
channel.
\end{proof}
\begin{corollary}
	 For almost every channel $\Phi$, the limit point
	 $\lim_{t\rightarrow\infty}e^{t\mathcal{L}_\Phi}$ is the completely depolarizing
	 channel, i.e. the channel that sends all states to the maximally mixed one.
\end{corollary}

\begin{definition}
	A channel written as $\Omega_t=\exp(t\c L)$, where $\c L=\Phi-{\1}$   is a
	time-independent Lindblad generator determined by a quantum channel $\Phi$ from
	a referred set $S$ of channels,  is called {\sl accessible}. The set of all
	accessible channels is denoted by $\c A$.
\end{definition}
There is no limitation on choosing the set $S$ of channels; for example, it can be 
the set of all channels, all unital channels, or a subset of channels determined 
based on what one can apply in a lab. Here we demand that the set $S$ be the 
semigroup formed by the convex hull of a group of channels. Hence the set $S$ here 
consists of mixed unitary channels; therefore, all elements are unital. \\

Let $G$ be a group consisting of $g=|G|$ quantum channels acting on states in
$\c H_N$, which means that all $\Phi\in G$ are unitary quantum channels, i.e.
$\Phi=U\otimes\overline{U}$. As a convention, let us take $\Phi_0=\1$ as the
identity (neutral) element of the group. We define the corresponding set of
Lindblad generators as $F_\c L(G):=\{\c L_{\Phi}=\Phi-\1 \}$ where  $\Phi\in \,
G-\{\1\}$. Note that we have excluded by hand the useless generator
corresponding to the identity from the set. Now, we can investigate the
accessibility problem, i.e. ask which quantum channels obtained by a convex
combination of the elements of $F_\c L(G)$ belong to the set $\cal A$ of
accessible channels. As the following theorem shows,  the  set $\cal A$ is a
subset of  the convex hull of the group elements $\sc C(G):= \{\Phi=\sum
p_\mu\Phi_\mu\}$ with $\sum_{\mu}p_{\mu}=1$, which is a semigroup itself. \\

\begin{theorem}\label{group}
Consider an identity map, $\Phi_0=\1$, 
and  a set of maps $\Phi_\mu$ which form a group $G$.
Then the dynamical semigroup $\e^{t\c L}$ generated by
$\c L=\sum_{\mu=1}^{g-1} q_\mu\c L_{\Phi_\mu}$,
 where $\c L_{\Phi_\mu}=\Phi_\mu-\1$,
is a convex combination of  the group members,
\begin{equation}\label{w_t}
	\e^{t\c L}=\sum_{\mu=0}^{g-1}w_\mu(t)\Phi_\mu,
\end{equation}
where the weights $w_{\mu}$ depend on time.
\end{theorem}


\begin{proof}
	Consider 
	a power series,
	\begin{equation}\label{grouptrajectory}
		\e^{t\c L}=\e^{t\left(\sum q_\mu\Phi_\mu-\1\right)}=
		\e^{-t}\sum_{m=0}^{\infty}
		\frac{\left(\sum_{\mu} tq_\mu\Phi_\mu\right)^m}{m\ !}.
	\end{equation}
	Closure of the group guarantees that the last expression on the right-hand side
	can be expanded based on group members with non-negative coefficients.
	Additionally, the left-hand side of this equation describes a trace-preserving
	quantum map at each moment in time. The above implies that the coefficients form
	a probability vector which completes the proof.
\end{proof}

The explicit time-dependent form of the probabilities $w_\mu (t)$ entering Eq.
(\ref{w_t}) relies on the group structure and not the specific representation of
the group $G$. Therefore, any faithful representation with any dimension, which
may not even present a quantum channel, can be adapted to compute the weights
$w_\mu(t)$. For example, one may think about the {\sl regular representation} of
a group. We remind the reader that the regular representation is written based
on the standard form of the Cayley table of the group. In this special
representation we assign a permutation matrix of dimension $g$, denoted by
$R_\alpha$, to each element of the group $G$ with the same cardinality such that
for any $\alpha$ and $\beta$ we have $\Tr\left(R_\alpha
R^\dagger_\beta\right)=g\delta_{\alpha,\beta}$, see Example 8 of \cite{SSMZK21}.
Hence, $w_\mu(t)$ can be obtained through the following lemma, which is a
consequence of the orthogonality of the permutations $R_\alpha$.
\begin{lemma}\label{regular}
	Let $G_\Phi=\{\Phi_0=\1,\cdots,\Phi_{g-1}\}$ be a representation of the group
	$G$ in terms of quantum channels. The group $G$ also admits a regular
	representation based on orthogonal permutations of dimension $g$ as
	$G=\{R_0=\1,\cdots,R_{g-1}\}$. An accessible quantum channel is  defined by Eq.
	\eqref{w_t}, in which
	\begin{equation}
		w_\mu(t) =
		\frac1g\Tr\left(R^\dagger_\mu \exp\left[t\left(\sum_{\nu=1}^{g-1}q_\nu 
		R_\nu-\1\right)\right]\right) = 
		\frac{\e^{-t}}{g}\Tr\left(R^\dagger_\mu 		
		\exp\left[t\sum_{\nu=1}^{g-1}q_\nu R_\nu\right]\right).
	\end{equation}
\end{lemma}

Additionally, as the group, $G_\Phi$ is assumed to be finite with unitary
elements $\Phi_\mu$, a convex polytope with exactly $g$ extreme points in the
set of all quantum channels can be constructed from them. Theorem \ref{group}
tells us the set of accessible maps also belongs to this polytope. The question
concerning the position of the trajectory $\e^{t\c L}$ in the polytope of the
group is highly related to the group structure, its subgroups, and the non-zero
interaction times (non-zero $q_\mu$). However, the trajectory is in the interior
of the polytope once all $q_\mu$ are non-zero. Moreover, it tends to the centre
of the polytope as $t\to\infty$.

\begin{proposition}\label{infinty}
	Let $G_\Phi=\{\Phi_\mu\}_{\mu=0}^{g-1}$ be a group of $g$ unitary quantum
	channels. The trajectory $\e^{t\c L}$ generated by a generic generator $\c
	L=\sum q_\mu\Phi_\mu-\1$ ends in the center of the polytope formed by the maps
	$\Phi_\mu$, i.e. the uniform mixture of all group members
	$\Phi_\ast=\frac1g\sum_{\mu=0}^{g-1}\Phi_\mu$.
\end{proposition}
\begin{proof}
It has been shown in Proposition \ref{limit} that the trajectory $\e^{t\c L}$
ends in the projector $\1_{\mathbb{V}_0}$ that projects on $\mathbb{V}_0$, the
invariant subspace of $\Phi=\sum q_\mu\Phi_\mu$, i.e.
\begin{equation}
\mathbb{V}_0 \equiv 
\Big\{
	\bm{x}\,\,\big|\,\,\Phi_\mu\bm{x}=\bm{x},\,\,\,\forall \Phi_\mu\in G_\Phi
\Big\}.
\end{equation}
Now write $\mathbb{H}$ as a direct sum of irreducible representations as
\begin{equation}
\mathbb{H} = \bigoplus_k \mathbb{V}_k.
\end{equation}
On the other hand, note that for all $\mu$ we get
$\Phi_\mu\Phi_\ast = \Phi_\ast\Phi_\mu=\Phi_\ast$ due to group 
rearrangement theorem. Applying Schur's lemma, we get
\begin{equation}
\Phi_\ast=\bigoplus_k c_k\1_{\mathbb{V}_k}.
\end{equation}
The equality $\Phi_\mu\Phi_\ast=\Phi_\ast$ implies $c_k=0$ for all 
$k\neq0$ and for $k=0$, it is possible to directly verify $c_0=1$. 
This shows $\Phi_\ast$ is the projector to the intersection of the invariant subspace of all group members and completes the proof.
\end{proof}

In the case with some vanishing weights, $q_\mu$, the situation will be
completely different. However, one gets the salient result from Theorem
\ref{group} that the trajectory remains in the subset characterized by the
smallest subgroup containing all non-vanishing quantum channels appearing in the
Lindblad generator. Such a curve ends in the centre of the polytope formed by
the smallest subgroup as well. As an example, let us remind readers the group
properties imply that every element $\Phi_\mu\in G_{\Phi}$ has order equal to
the smallest integer and the positive number $h_\mu$ such that
$\Phi_\mu^{h_\mu}=\1$. The trajectory generated by $\c L_{\Phi_\mu}=\Phi_\mu-\1$
is defined by
\begin{equation}\label{1lindblad}
	\e^{t_\mu\c L_{\Phi_\mu}}=\e^{t_\mu(\Phi_\mu-\1)}=
	\e^{-t_\mu}\e^{t_\mu\Phi_\mu}=
	\e^{-t_\mu}\sum_{l=0}^{\infty}\frac{(t_\mu\Phi_\mu)^l}{l\ !}=
	\e^{-t_\mu}\left(\sum_{l\in M^\mu_0}\frac{t_\mu^l}{l\ !} \ \1+
	\sum_{l\in M^\mu_1}\frac{t_\mu^l}{l\ !}\ \Phi_\mu+
	\dots+\sum_{l\in M^\mu_{h_\mu-1}}\frac{t_\mu^l}{l\ !}\ 
	\Phi_\mu^{h_\mu-1}\right).
\end{equation}
Here $M^\mu_r$ is the set of all non-negative numbers congruent modulo $h_\mu$,
whose remainder, when divided by $h_\mu$, is equal to $r$. 
Note that
\begin{eqnarray}\label{prob}
	p_r(t_\mu):=
	\e^{-t_\mu}\sum_{l\in M^\mu_r}\frac{t_\mu^l}{l\ !}=
	\e^{-t_\mu}\sum_{n=0}^{\infty}
	\frac{t_\mu^{nh_\mu+r}}{(nh_\mu+r)\ !}=
	\frac{\e^{-t_\mu}}{h_\mu}\sum_{s=0}^{h_\mu-1}
	\omega_{\mu}^{-sr}\exp[t_\mu\omega_{\mu}^{s}],
\end{eqnarray}
where $\omega_\mu=\exp[\frac{2\pi i}{h_\mu}]$. Provided that $t_\mu>0$, these
coefficients are all positive numbers adding up to $1$, i.e. they provide a
probability vector of length $h_\mu$ and none of its elements vanishes. Thus we
can write
\begin{equation}\label{oneterm}
	\e^{t_\mu\c L_{\Phi_\mu}}=\sum_{r=0}^{h_\mu-1}p_r(t_\mu)\Phi_\mu^r.
\end{equation}
which means at each moment of time we get a mixture of all $h_\mu$ different
powers of $\Phi_\mu$. This implies $\e ^{t_\mu\c L_{\Phi_\mu}}$ (with no
summation on $\mu$), for any $t_\mu>0$, is located in the interior of the 
polytope formed by taking the convex hull of $h_\mu$ extreme points each of
which is assigned to different powers of $\Phi_\mu$. Moreover, at
$t_\mu\to\infty$ the trajectory tends to the center of the polytope as discussed
in Proposition \ref{infinty}, i.e. the uniform mixture of all different powers
of $\Phi_\mu$. To see it explicitly, we should note that Eq.~\eqref{prob} can 
also be written as
\begin{equation}
	p_r(t_\mu)=\frac{1}{h_\mu}\left(1+\sum_{s=1}^{h_\mu-1}
	\omega_\mu^{-sr}\e^{-2t_\mu\sin^2(\pi s/{h_\mu})}
	\e^{-it_\mu\sin(2\pi s/h_\mu)}\right).
\end{equation}
This equation now shows that at large time scale
$\lim_{t_\mu\to\infty}p_r(t_\mu)=\frac{1}{h_\mu}$ holds  for any $r$.\\

We emphasise that Eq. \eqref{prob} is independent of the dimension and the
individual channels that form the group. It only depends on the order of the
element under investigation. Moreover, note that the set
$H_{\Phi_\mu}:=\{\Phi_\mu^r\}$ for $r\in\{0,\dots,h_\mu-1\}$ is actually a
cyclic subgroup of  $G_\Phi$ of order $h_\mu$. For an arbitrary convex
combination of Lindblad generators associated with different elements of
$H_{\Phi_\mu},$ the resultant trajectory cannot leave the polytope formed by the
elements of $H_{\Phi_\mu}$.\\

In other words, if all quantum channels appearing in $\exp[t(\sum q_{\mu}
\Phi'_{\mu} -\1)]$ belong to a cyclic subgroup, then at each moment of time, the
trajectory can be convexly expanded in terms of all members of the smallest
cyclic subgroup containing all $\Phi'_{\mu}$ in this summation. Due to
commutativity of the elements of a cyclic group, we get
$\exp{(t\c L)}=\exp{\left(\sum t_{\mu}\c L_{\Phi'_{\mu}}\right)}= 
\prod\exp({t_{\mu}\c L_{\Phi'_\mu}})$. 
Each term in the last formula can be expanded based on Eq. \eqref{oneterm} in
which probabilities are given by Eq. \eqref{prob}. To see this fact explicitly,
let $\Phi_{\mu}$ be the generator of $H_{\Phi_{\mu}}$, i.e.
$H_{\Phi_{\mu}}=\{\Phi_{\mu}^i\}_{i=0}^{h_{\mu}-1}=
\{\Phi'_{\mu}\}_{\mu=0}^{h_{\mu}-1}$, then
\begin{eqnarray}\label{cycliclindblad}
	\exp\left(\sum t_{\mu}\c L_{\Phi'_{\mu}}\right)&=&\nonumber
	\exp\left(\sum_{i=1}^{h_\mu-1} t_{\mu_i}\c L_{\Phi_{\mu}^i}\right)=
	\prod_{i=1}^{h_\mu-1}\exp\left(t_{\mu_i}\c L_{\Phi_\mu^i}\right)\\&=&
	\prod_{i=1}^{h_\mu-1}
	\left(\sum_{j=0}^{h_{\mu_i}-1}p_j(t_{\mu_i})\Phi_{\mu}^{ij}\right)
	=\sum_{l}\mathzapf{P}_l(t_{\mu})\Phi_{\mu}^l,
\end{eqnarray}
where $h_\mu$, and $h_{\mu_i}$ are the order of $\Phi_\mu$, and $\Phi_\mu^i$, 
respectively. In the last equation $l=(ij\ {\rm mod}\ h_\mu)$ and 
\begin{equation}\label{prob2}
\mathzapf{P}_l(t_{\mu})\equiv \underset{\stackrel{\scriptstyle {\rm such\ that:}}{(ij\ {\rm mod}\ h_{\mu})=l}}{\prod_{i=1}^{h_\mu-1}
\sum_{ j=0}^{h_{\mu_i}-1}}p_j(t_{\mu_i}).
\end{equation}
Hence, we can find the time-dependent probabilities in the convex combination.
However, this is only one of the possibilities since, in a general case, the
polytope formed by the set elements $G_\Phi$ can differ form the simplex.\\

Consider now another quantum channel $\Phi^{\prime\prime}_\alpha$ from the set
$G_\Phi$ such that there is no cyclic subgroup of $G_\Phi$ including both
$\Phi_\mu$ and $\Phi^{\prime\prime}_\alpha$. To emphasize this fact, we will
denote $\Phi^{\prime\prime}_\alpha$ by $\Theta_\alpha$. If we want to add $\c
L_{\Theta_\alpha}$ to the set discussed in Eq.~(\ref{cycliclindblad}) two
different cases are possible. First, if we assume
$[\Phi_{\mu},\Theta_\alpha]=[\c L_{\Phi_\mu},\c L_{\Theta_\alpha}]=0$, then
\begin{eqnarray}
\exp\left(t_\alpha\c L_{\Theta_\alpha}+\sum t_{\mu}\c L_{\Phi'_{\mu}}\right)&=&
\exp\left(t_\alpha\c L_{\Theta_\alpha}\right)
\exp\left(\sum_{i=1}^{h_\mu-1} t_{\mu_i}\c L_{\Phi_{\mu}^i}\right)
=\sum_{k,l}p_k(t_\alpha)
\mathzapf{P}_l(t_{\mu})\Theta_\alpha^k\Phi_{\mu}^l,
\end{eqnarray}
where $h_\alpha$ and $h_\mu$ are the order of $\Theta_\alpha$ and $\Phi_\mu$,
and $p_k(t_\alpha)$ and $\mathzapf{P}_l(t_\mu)$ are introduced in Eqs.
\eqref{prob} and \eqref{prob2}, respectively. Note that commutativity of
$\Theta_\alpha$ and $\Phi_\mu$ imposes compatibility on any power of them, i.e.
all elements of the cyclic subgroups $H_{\Theta_\alpha}$ and $H_{\Phi_\mu}$
commute. Moreover, as we assumed that there is no cyclic subgroup including both
$\Theta_\alpha$ and $\Phi_\mu$, thus $H_{\Theta_\alpha}$ and $H_{\Phi_\mu}$ have
no element in common but the identity. This implies that the product of
$H_{\Theta_\alpha}$ and $H_{\Phi_\mu}$ is a subgroup of $G_\Phi$ of order
$h_\alpha h_\mu$, i.e. $\forall k,l:\ \ \Theta_\alpha^k\Phi_\mu^l$ belongs to a
subgroup of $G_\Phi$ formed by multiplication of its cyclic and commutative
subgroups.
Applying the above arguments, we can generalize these results as follows.


\begin{corollary} \label{2cyclic lindblad}
Let quantum channels $\Theta$ and $\Phi$ respectively denote the generators of
the cyclic groups $H_\Theta$ of order $h_\Theta$ and $H_\Phi$ of order $h_\Phi$.
Assume that $[\Theta,\Phi]=0$ and that $H_\Theta$ and $H_\Phi$ have no element
in common but the identity. Consider any convex combination of the elements of
$\{\Theta^i-\1\}_{i=1}^{h_\Theta-1}\bigcup \{\Phi^j-\1\}_{j=1}^{h_\Phi-1}$.

Then the dynamical semigroups generated by such Lindblad generators belong to
the polytope constructed by the convex hull of $h_\Theta h_\Phi$
extreme points related to different elements of $H_{\Theta}\times H_{\Phi}$. The
trajectory, $\exp(t {\c L})$,  starts from $\Theta^0=\Phi^0=\1$ when all
interaction times $t_{\Theta^i}$ and $t_{\Phi^j}$ are zero,  it  goes inside the
polytope if there exist $i$ and $j$   
for which $t_{\Theta^i}\neq0$ and $t_{\Phi^j}\neq0$, and it tends to the
projector at the centre of polytope at $t\to\infty$.
\end{corollary}


\begin{corollary}\label{abelian accessibility}
Consider an abelian group $G_\Phi$ represented by  the maps $\Phi_\mu$ and  the
Lindblad generator of the form
$\c L=
\sum_\mu q_\mu\Phi_\mu-\1$.
Then the dynamical semigroup $\e^{t\c L}$, generated by $\c L$, can be
decomposed convexly based on the elements of a subgroup of $G_\Phi$ obtained by
the product of different cyclic subgroups, each formed by powers of all the maps
$\Phi_\mu$ appearing in $\c L$.
\end{corollary}

This Corollary can be proved through the fundamental theorem of abelian groups and 
the above discussion.\\

The condition on commutativity in Corollary \ref{2cyclic lindblad} may be
relaxed in a general scenario, i.e. $[\Theta,\Phi]\neq0$. In that case,
$H_{\Theta}\times H_{\Phi}$ is not necessarily a subgroup. Therefore, we need to
add some other quantum channels from $G_\Phi$ to the set $H_{\Theta}\times
H_{\Phi}$ so it satisfies closure. The smallest subgroup $H$ of $G_\Phi$ that
includes all members of the set $H_{\Theta}\times H_{\Phi}$ expands convexly the
dynamical semigroups generated by Lindblad generators of our interest. This is a
direct consequence of the closure property, as  the group structure plays a
crucial role for Theorem \ref{group}.
\section{Exemplary groups of channels and corresponding sets of accessible maps}
\label{Examples}

In this section we will apply the aforementioned theorems and relations for
certain constructive examples. For some groups associated with quantum channels
acting on $N$-dimensional systems and represented by matrices of size $N^2$ we
demonstrate how to find the quantum channels, their polytopes, and the subset of
accessible maps inside the polytope.\\

Before proceeding with the examples let us mention that if
$G_\Phi=\{\Phi_\mu=U_\mu\otimes\overline{U}_\mu\}_{\mu=0}^{g-1}$ is a group of
$g$ unitary channels, then the set of the corresponding Kraus operators denoted
by $\tilde{G}_U=\{U_\mu\}_{\mu=0}^{g-1}$ is a {\sl group up to a phase}, which
is also called a {\sl projective representation} of group $G$. Explicitly, in
the set  $\tilde{G}_U$:
(i) there are no two elements  which are equal up to a phase,
i.e. if $U\in \tilde{G}_U$, then $\e^{i\phi}U\notin \tilde{G}_U$,
(ii)  the set is closed up to a phase, i.e. if $U_1,U_2\in \tilde{G}_U$, then
$\e^{i\phi}U_1U_2\in \tilde{G}_U$ for some $\phi$,
(iii) and it contains an inverse up to a phase for each element, i.e. for any 
$U_1\in \tilde{G}_U$ there exists an element $U_2\in \tilde{G}_U$ and a phase 
$\phi$ such that $U_1U_2=\e^{i\phi}\1$.
Clearly, the last two conditions imply that (iii') the set possesses a neutral
element up to a phase. Moreover, the set $G_\Phi$ is abelian if and only if
$\tilde{G}_U$ is abelian up to a phase, i.e. $\forall U_1,U_2\in \tilde{G}_U$
one has $U_1U_2=\e^{i\phi} U_2U_1$ for some phase $\phi$.

\begin{example}[\bf The group of order $g=2$]\normalfont \label{g=2}

The first non-trivial example is a group of order $2$,
$G_\Phi=\mathsf{Z}_2=\{\Phi^{\sf Z_2}_0=\1,\Phi^{\sf Z_2}_1\}$ which  is a cyclic 
group. The superscript $\sf Z_2$ is to emphasis that the group is a cyclic group of 
order $2$.
Since $\Phi^{\sf Z_2}_1$ is of group order $h_1=2$, its spectrum contains only 
$\pm1$. 
In order to have two distinguished elements in the group,
there has to be $-1$ in the spectrum. Indeed, for an $m$ from the set 
$\{1,\dots,[N/2]\}$,  $2m(N-m)$ out of $N^2$  eigenvalues  of $\Phi^{\sf Z_2}_1$ 
may be equal to $-1$. To see that, let us denote by $\tilde{G}_{U}^{\sf 
Z_2}=\{U_0^{\sf Z_2}=\1_N,U_1^{\sf Z_2}\}$ the group up to a phase of Kraus 
operators. 
To have two distinguished elements in $\tilde{G}_{U}^{\sf Z_2}$ at least one of the 
eigenvalues of $U_1^{\sf Z_2}$ should possess a phase difference equal to $\pi$ 
from other eigenvalues which are equal. However, there might be $m$ eigenvalues 
with such a phase difference in 
the spectrum in the general case where $m$ belongs to the aforementioned set. Note 
that $2m(N-m)$ negative eigenvalues will appear in the spectrum of the quantum 
channel  $\Phi_1^{\sf Z_2}=U_1^{\sf Z_2}\otimes\overline{U}_1^{\sf Z_2}$.\\

The polytope related to the group, $G_\Phi=\sf Z_2$, is the simplex of dimension 
one, i.e. a line segment. 
The set of accessible maps, in this case, can be written as \eqref{oneterm} with 
probabilities given by Eq. \eqref{prob}
\begin{equation}\label{c2}
	\Omega_t=\frac{1+\e^{-2t}}{2}\1+
	\frac{1-\e^{-2t}}{2}\Phi^{\sf Z_2}_1.
\end{equation}
Hence the trajectory starts from the identity and goes to the projector at the
centre of the line at $t\rightarrow\infty$, see Fig.\ref{two-dim}(a). Thus, any
map of the form $\Phi=w\1+(1-w)\Phi_1^{\sf Z_2}$ with $w\geq1/2$ is accessible.
Note that the results do not depend on the dimension of the quantum channels,
and such a simplex exists in all dimensions. As an explicit example, however,
one may get the qubit channel $\Phi_1^{\sf Z_2}=\sigma_i\otimes\sigma_i$ with
one of the Pauli matrices.
\end{example}

\begin{example} [\bf The group of order $g=3$]\normalfont
A group of order $3$ is still a cyclic and so an abelian group, $G_{\Phi}=\sf
Z_3=\{\Phi_0^{\sf Z_3}=\1,\Phi_1^{\sf Z_3},\Phi_2^{\sf Z_3}\}$, where the
superscript $\sf Z_3$ denotes the cyclic group of order $3$. Since all the
members of this group are unitary channels and so the extreme points of the set
of quantum channels, these channels do not lie on the same line. Thus, they are
linearly independent and the polytope of the group is a triangle for
$N$-dimensional quantum channels. This triangle is regular with respect to the
Hilbert-Schmidt distance between two unitary maps $\Phi_i$ and $\Phi_j$,
\begin{equation}\label{HSdistance}
D(\Phi_i,\Phi_j)=\|\Phi_i-\Phi_j\|_2^2=2N^2-\Tr(\Phi_i\Phi_j^\dagger)-
\Tr(\Phi_i^\dagger\Phi_j).
\end{equation}
In the case of the group of order $3$  we get
\begin{equation}
\forall i\neq j:\qquad D(\Phi_i^{\sf Z_3},\Phi_j^{\sf Z_3})=
2N^2-\Tr(\Phi_1^{\sf Z_3})-\Tr(\Phi_2^{\sf Z_3}),
\end{equation}
showing all elements of the group are equidistant. Therefore, the triangle is 
regular for  $N$-dimensional quantum channels, and we deal with a 2-simplex.
On the other hand, the order of both non-trivial elements are three, 
$h_{1}=h_{2}=3$. Thus, through Eq. \eqref{cycliclindblad} and Eq. \eqref{prob2}, 
accessible maps are given by $\Omega_t=w_0(t)\1+w_1(t)\Phi_1^{\sf 
Z_3}+w_2(t)\Phi_2^{\sf Z_3}$, where $t=t_1+t_2$ and
\begin{equation}
	\begin{split}		
	w_0(t) &= \frac{1}{3}\left(1+2 e^{\frac{-3}{2} (t_1+t_2)}
	\cos \left(\frac{\sqrt{3}}{2}  (t_1-t_2)\right)\right),\\
	w_1(t) &= \frac{1}{3}\left(1- e^{\frac{-3}{2} (t_1+t_2)}
	\cos \left(\frac{\sqrt{3}}{2}  (t_1-t_2)\right)+
	\sqrt{3} e^{\frac{-3}{2} (t_1+t_2)}
	\sin \left(\frac{\sqrt{3}}{2}  (t_1-t_2)\right)\right),\\
	w_2(t) &= \frac{1}{3}\left(1- e^{\frac{-3}{2} (t_1+t_2)}
	\cos \left(\frac{\sqrt{3}}{2}  (t_1-t_2)\right)-
	\sqrt{3} e^{\frac{-3}{2} (t_1+t_2)}
	\sin \left(\frac{\sqrt{3}}{2}  (t_1-t_2)\right)\right).
	\end{split}
\end{equation}
These equations are already introduced in \cite{SAPZ20} in the study of
accessible Weyl channels of dimension three. The simplex representing  this
group and the subset $\cal A$  of accessible maps are shown in Fig.
\ref{two-dim}b. Since the group has no non-trivial subgroup, there is no
accessible map on the edges, as shown in this figure. Observe that this result
is independent of the dimension of the quantum channels. Again, as an explicit
example we may think of a group of three single-qubit unitary channels based on
the group up to a phase,  $\tilde{G}_U^{\sf Z_3}=\{\1_2,\;
\diag[\e^{\frac{i\pi}{3}},\e^{-\frac{i\pi}{3}}], \;
\diag[\e^{-\frac{i\pi}{3}},\e^{\frac{i\pi}{3}}]\}$. 
Existence of  a faithful representation of qubit channels for this group
guarantees the existence of a faithful representation in higher dimensions. Note
that the same non-convex subset of the equilateral triangle, plotted in Fig.
\ref{two-dim}b and bounded by two logarithmic spirals  \cite{SAPZ20}, was
earlier identified in search of classical semigroups \cite{LRRL16,KL21} in the
space of bistochastic matrices of order three.\\

\begin{figure}
\centering
\includegraphics[scale=0.35]{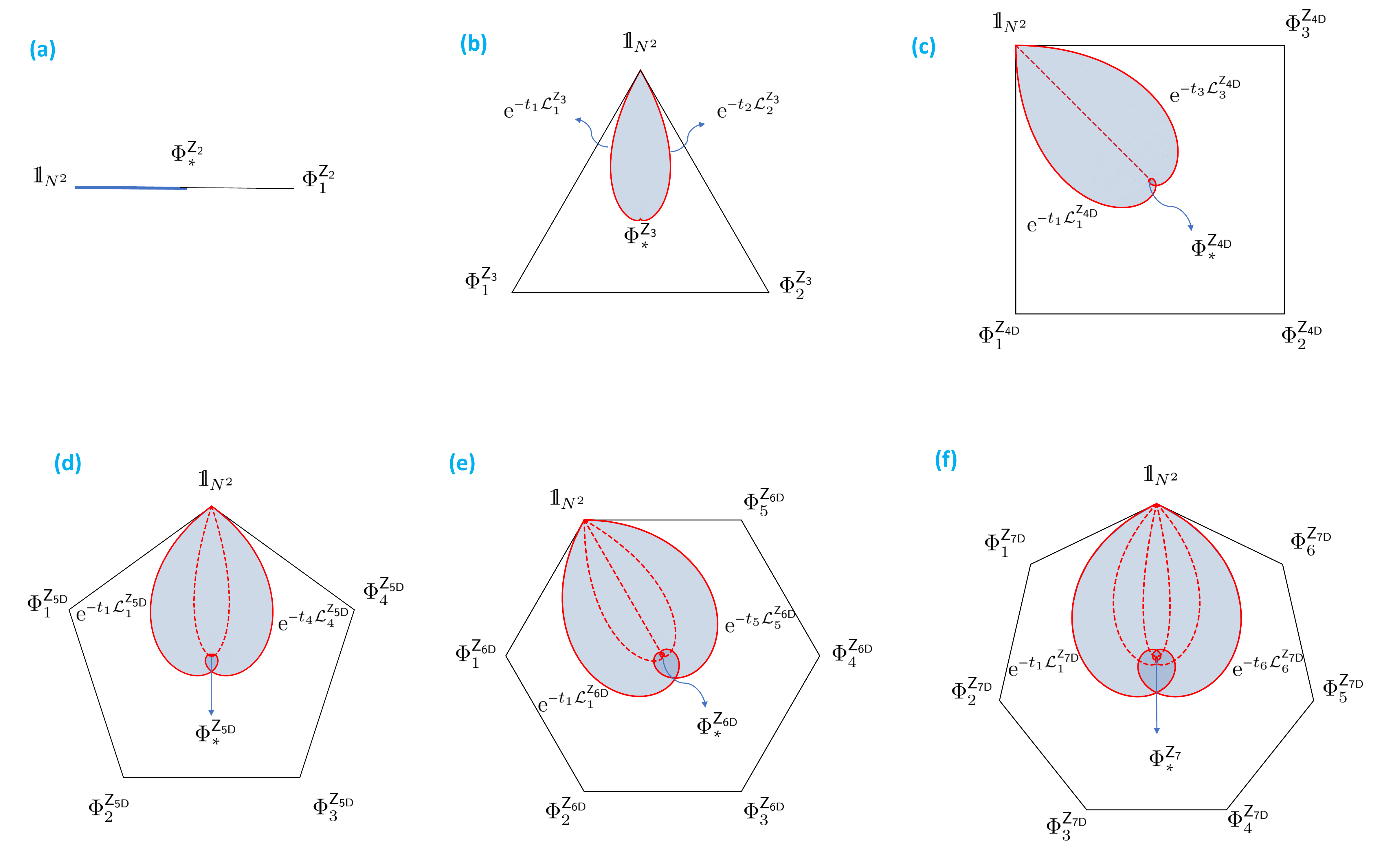}
\caption{The set $\cal A$ of accessible maps, indicated by blue, in
(a) 1-simplex of a cyclic group of oreder $2$,
(b) 2-Simplex of a cyclic group of order $3$, 
and the set of channels in the regular polygons of linearly dependent channels 
forming a cyclic group of order (c) $4$, (d) $5$, (e) $6$, and (f) $7$.
The symbol $\1_{N^2}$ denotes the identity map and the subscript $N^2$ is added
to emphasize that it can be considered in any dimension $N$. The map $\Phi_\ast$
indicates the center of each polygon. All sets $\cal A$ are star-shaped.
}
\label{two-dim}
\end{figure}

Concerning the spectrum, the fact that non-trivial elements are of order $3$
implies that eigenvalues of these maps are $E_\Phi=\{1,\e^{\pm i 	
\frac{2\pi}{3}}\}$. Perron-Frobenius theorem in the space of quantum channels
suggests all of these elements are present in the spectrum of both $\Phi_1^{\sf
Z_3}$ and $\Phi_2^{\sf Z_3}$. Due to the commutativity of these maps, the
spectra of channels of the form $\Phi_{\vec{p}}=\sum p_i\Phi_i$ in the complex
plane are embedded in the regular triangle whose vertices are given by the
elements of the set $E_\Phi$ above, see Fig.~\ref{spectra}(a). Observe that  the
spectra of accessible maps in this the triangle form the same shape bounded by
logarithmic spirals  as the set $\cal A$ of accessible maps in the simplex of
the group in the space of quantum channels -- see Fig.~\ref{two-dim}(b).\\
\end{example}

\begin{example}[\bf The groups of order $g=4$]\normalfont
A group of order $4$ is isomorphic either to a cyclic group $G_\Phi=\sf
Z_4=\{\Phi_0^{\sf Z_4}=\1,\Phi_1^{\sf Z_4}, \Phi_2^{\sf Z_4},\Phi_3^{\sf Z_4}\}$
or to a non-cyclic one, $G_\Psi=\{\Psi_0=\1,\Psi_1,\Psi_2,\Psi_3\}$ with
$\Psi_i^2=\1$ and $\Psi_i\Psi_j=\Psi_k$ for different $i,j,k\in\{1,2,3\}$. Both
of these groups are still abelian so one can find a set of common basis in which
all the elements are diagonal. Let us first consider the cyclic group, $G_\Phi$.

\subsubsection{The cyclic group of order $g=4$}\label{cyclic}

The polytopes of this group can be two-dimensional or three-dimensional related
to the linear dependency of the elements. The spectra of the group can determine
this up to a phase of Kraus operators.

\begin{lemma}\label{LinDep}
	Consider the cyclic group  $G_\Phi=\sf Z_4$ of order $4$ of quantum channels 
	acting on $N$-dimensional systems. Let 
	$\Phi_1^{\sf Z_4}=U_1^{\sf Z_4}\otimes \overline{U_1^{\sf Z_4}}$ 
	denote its generator, i.e. $\Phi_i^{\sf Z_4}=(\Phi_1^{\sf Z_4})^i$ for
	$i\in\{0,\cdots,3\}$. So $U_1^{\sf Z_4}$ is the generator of a cyclic group up
	to a phase of order $4$, $\tilde{G}_U^{\sf Z_4}$, whose spectrum up to a general
	phase is from the set $E_U=\{\pm1,\pm i\}$. Moreover, both real and imaginary
	elements should be present in the spectrum of $U_1^{\sf Z_4}$ to have four
	distinguished elements in the group $G_\Phi$. Then the elements of $G_\Phi=\sf
	Z_4$ are linearly dependent if and only if all the real elements appearing in
	the spectrum of $U_1^{\sf Z_4}$ have the same sign and all the imaginary
	eigenvalues have also the same sign (can be different from the sign of real
	ones).
\end{lemma}

\begin{proof}
	For any unitary quantum channel $\Phi=U\otimes\overline{U}$, eigenvalues of  the
	superoperator $\Phi_U$ are determined  by  the eigenvalues of $U$ through
	$\lambda_{ij}(\Phi)=\lambda_i(U)\overline{\lambda}_j(U)$. Moreover, conditions
	on the spectrum of $U_1^{\sf Z_4}$ imply that the eigenvalues of 
	$\Phi_1^{\sf Z_4}$ read $E_{\Phi_1^{\sf Z_4}}=\{1,\pm i\}$. 
	Hence in the spectra of $\Phi_2^{\sf Z_4}$ and $\Phi_3^{\sf Z_4}$ corresponding
	to each $\pm i\in E_{\Phi_1^{\sf Z_4}}$ the numbers $-1$ and $\mp i$ appear,
	respectively. Accordingly, we get 
	$\Phi_0^{\sf Z_4}+\Phi_2^{\sf Z_4}=\Phi_1^{\sf Z_4}+\Phi_3^{\sf Z_4}$ 
	meaning that the group members are linearly dependent and the polytope of the
	group lies on a plane, see Fig. \ref{two-dim}(c). On the other hand, once the
	condition on the spectrum of $U_1^{\sf Z_4}$ is violated, then $\pm 1$ and $\pm
	i$ appear in the spectrum of $\Phi_1^{\sf Z_4}$. In this case, one can check
	that the equation $\sum c_i \Phi_i=0$ implies $c_i=0$ for any $i$, suggesting
	the the group elements are linearly-independent, and the polytope is then a
	tetrahedron.
\end{proof}

\begin{figure}
\centering
\includegraphics[scale=0.35]{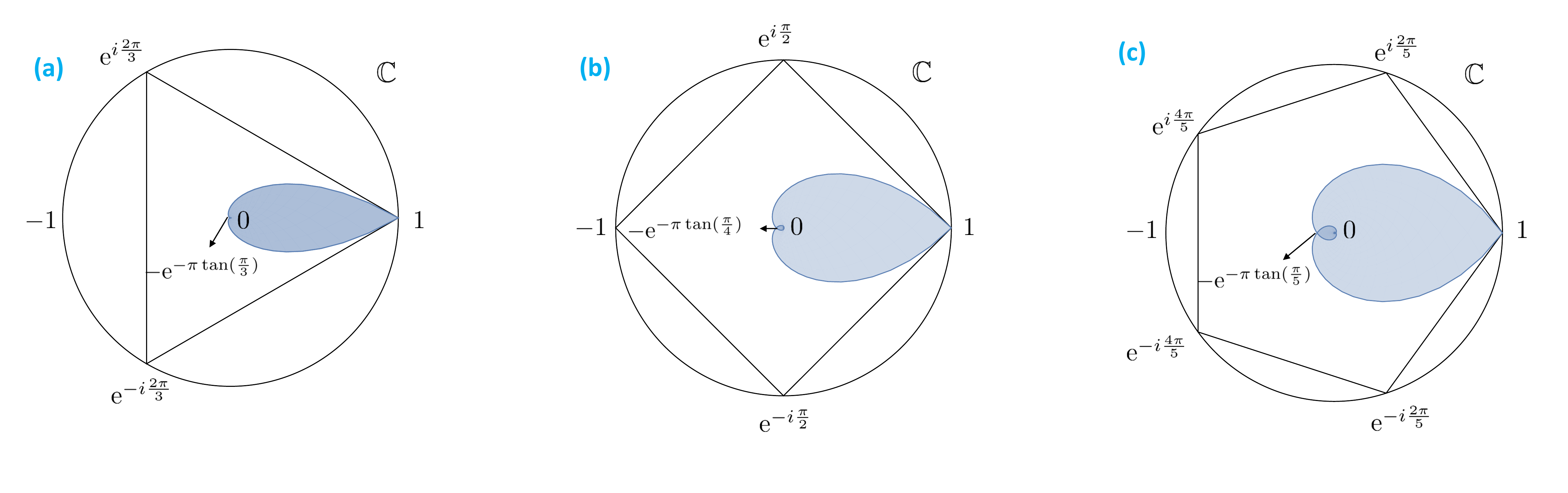}
\caption{The regular polygons in the complex plane represent 
support of spectra of the  quantum channels forming a cyclic group of order (a)
$3$, (b) $4$, and (c) $5$. The blue region in each polygon determines the
support of spectra of the accessible maps in the corresponding set. Note that
the shape and boundaries of the blue regions above are exactly the same as the
quantum channels in Fig.~\ref{two-dim}(b)-(d).
}
\label{spectra}
\end{figure}
 
We will change the superscript $\sf Z_4$ as $\sf Z_{4D}$  and $\sf Z_{4I}$ for
the linearly dependent and linearly independent case, respectively, when we need
to emphasise this property. An immediate result of the above lemma is the
following corollary.

\begin{corollary}\label{g4qubit}
	For the cyclic group of order $4$ formed by qubit channels, acting on systems 
	of dimension $N=2$, elements are linearly dependent. 
\end{corollary}

However, in higher dimensions both two and three dimensional polytopes are
possible. In order to study the regularity of the poytope, let $r_+$ and $r_-$
denote the number of eigenvalues $+1$ and $-1$, respectively,  in the spectrum
of $U_1^{\sf Z_4}$, defined in Lemma \ref{LinDep}. In a similar way let $i_+$
and $i_-$ denote the number of eigenvalues $+i$ and $-i$, so that
$r_++r_-+i_++i_-=N$. These numbers allow us to express the Hilbert-Schmidt
distance \eqref{HSdistance} between the channels,
\begin{equation}\label{dist-4}
	\begin{split}
		D(\Phi_0^{\sf Z_4},\Phi_1^{\sf Z_4})&=
			D(\Phi_0^{\sf Z_4},\Phi_3^{\sf Z_4})=
			D(\Phi_1^{\sf Z_4},\Phi_2^{\sf Z_4})=
			D(\Phi_2^{\sf Z_4},\Phi_3^{\sf Z_4})\\
		&=2N^2-\Tr\Phi_1^{\sf Z_4}-\Tr\Phi_3^{\sf Z_4}
			=8r_+r_-+8i_+i_-+4(r_++r_-)(i_++i_-),\\
		D(\Phi_0^{\sf Z_4},\Phi_2^{\sf Z_4})&=
			D(\Phi_1^{\sf Z_4},\Phi_3^{\sf Z_4})=
			2N^2-2\Tr\Phi_2^{\sf Z_4}=8(r_++r_-)(i_++i_-) .
	\end{split}
\end{equation}

The above relations imply if the polytope  is of dimension $2$, it forms a
square like the one in Fig.~\ref{two-dim}(c), while the three-dimensional case
can be a (non-)regular tetrahedron -- see Fig. \ref{4-cyc}. In particular, for a
cyclic group of order $4$ of linearly-independent qutrit channels one out of
four numbers $r_{\pm}$ and $i_{\pm}$ is zero due to  Lemma \ref{LinDep}, while
the remaining three numbers are equal to $1$. This implies that both distances
in Eq.~\eqref{dist-4} are equal. Hence, the polytope is a regular tetrahedron
for a cyclic group of order, $4$ formed by linearly independent quantum channels
acting on $N=3$ dimensional systems. In higher dimensions, however, both regular
and non-regular tetrahedrons, as well as the square, are possible, see
Fig.~\ref{4-cyc}. To emphasize the regularity of the polytope, we will denote
cyclic groups of four linearly-independent channels with a regular and
non-regular tetrahedrons by
$G_\Phi^{\sf Z_{4I}}=\{\Phi_0^{\sf Z_{4I}}=\1,\Phi_1^{\sf Z_{4I}},
\Phi_2^{\sf Z_{4I}},\Phi_3^{\sf Z_{4I}}\}$ and $G_\Gamma^{\sf Z_{4I}}=
\{\Gamma_0^{\sf Z_{4I}}=\1,\Gamma_1^{\sf Z_{4I}},
\Gamma_2^{\sf Z_{4I}},\Gamma_3^{\sf Z_{4I}}\}$, respectively. \\

The shape of the polytope formed by the elements of the group $G$ in the space
of quantum channels and the dimension of the channels do not play any role by
describing the boundary of the set $\cal A$ of  accessible maps. Hence the group
$G_{\Phi}^{\sf Z_4}$ denotes one of three different types of the group:
$G_\Phi^{\sf Z_{4D}}$, $G_\Phi^{\sf Z_{4I}}$, and $G_\Gamma^{\sf Z_{4I}}$. Due
to Corollary \ref{abelian accessibility}, the order of the cyclic subgroups (the
order of elements) plays a pivotal role. Assuming $\Phi_1^{\sf Z_4}$ as the
generator of $G_{\Phi}^{\sf Z_4}$, we get $h_{1}=h_{3}=4$ while $h_{2}=2$. This
implies that
\begin{eqnarray}
&\e^{t_1\c L_{\Phi_1}^{\sf Z_4}}=&\frac{\e^{-t_1}}{2}\left[
\left(\cosh t_1+\cos t_1\right)\1+
\left(\sinh t_1+\sin t_1\right)\Phi_1^{\sf Z_4}+
\left(\cosh t_1-\cos t_1\right)\Phi_2^{\sf Z_4}+
\left(\sinh t_1-\sin t_1\right)\Phi_3^{\sf Z_4}\right],
\end{eqnarray}
while $\e^{t_2\c L_{\Phi_2}^{\sf Z_4}}$ is the same as Eq. \eqref{c2} with
substitution of $t_2$ for $t$, and noting that $\Phi_{2}^{\sf Z_4}=\Phi_1^{\sf
	Z_2}$. The trajectory $\e^{t_3\c L_{\Phi_3}^{\sf Z_4}}$ is the same as
$\e^{t_1\c L_{\Phi_1}^{\sf Z_4}}$ by replacing $t_1$ by $t_3$ and $\Phi_1^{\sf
	Z_4}$ by $\Phi_3^{\sf Z_4}$. Hence for an accessible map of a cyclic group of
order $4$, we get $\Omega_t=\sum w_i(t)\Phi_i^{\sf Z_4}$ with $t=t_1+t_2+t_3$
independently  of the shape of its polytope and dimension of the channels. In
this way we arrive at equations
\begin{equation}\label{w_c4i}
	\begin{split}
	w_0(t)&=\frac14\left(1+\e^{-(t_1+t_3)}+2\e^{-2t_2}\cos(t_1-t_3)\right), \\
	w_1(t)&=\frac14\left(1-\e^{-(t_1+t_3)}+2\e^{-(t_1+2t_2+t_3)}
	\sin(t_1-t_3)\right),\\
	w_2(t)&=\frac14\left(1+\e^{-(t_1+t_3)}-2\e^{-2t_2} \cos(t_1-t_3)\right), \\
	w_3(t)&=\frac14\left(1-\e^{-(t_1+t_3)}-
	2\e^{-(t_1+2t_2+t_3)}\sin(t_1-t_3)\right),
	\end{split}
\end{equation}
which allow us to represent any element of the set $\cal A$ of accessible maps
as a convex combination of the group members. In the case of the group
$G_{\Phi}^{\sf Z_4D}$ the set $\cal A$ of accessible maps forms a non-convex
subset in the square -- see Fig.~\ref{two-dim}(c).

The spectra of superoperators obtained by a convex combination of the elements
of this group are also embedded in a square in the complex plane whose vertices
are $\{\pm1,\pm i\}$. The spectra corresponding to accessible maps are
restricted to a shape equivalent to the shape of accessible maps in the set of
quantum channels, see Fig. \ref{spectra}(b). As a simple example for such a
group, one may think of the set of four qubit channels constructed by the group
up to a phase $\tilde{G}_U^{\sf Z_{4D}}=\{\1_2,\;
\diag[\e^{\frac{i\pi}{4}},\e^{\frac{-i\pi}{4}}],\; \diag[i,-i],\;
\diag[\e^{\frac{-i\pi}{4}},\e^{\frac{i\pi}{4}}]\}$. 
Having a qubit realization, this group can also find representations in higher
dimensions. \\

Fig.~\ref{4-cyc} shows the subset of accessible maps in a regular and a
non-regular tetrahedron. As mentioned in Corollary~\ref{g4qubit}, it is
impossible to find such polytopes satisfying a cyclic group structure in the
space of qubit channels. Let us first consider the regular tetrahedron, Fig.
\ref{4-cyc} (a) and (b). The simplest example of a cyclic group of order four
with linearly-independent elements  is a set of qutrit channels based on the
group up to a phase,
$\tilde{G}_U^{\sf Z_{4I}}=\{\1_3,\; \diag[1,-1,i],\; \diag[1,1,-1],\; \diag[1,-1,-i]\}$.
Due to Lemma \ref{LinDep} existence of both $1$ and $-1$ in the spectrum of
these unitary operators implies that they are linearly independent. One can find
other examples of qutrit channels as well as channels in higher dimensions for
this group. However, the shape of accessible maps inside the tetrahedron is
independent of a particular example or even dimension of the channels.\\

We shall now analyze  the case of a non-regular tetrahedron, presented in Fig.
\ref{4-cyc} (c) and (d). This special non-regular tetrahedron is the polytope of
a cyclic group of four linearly-independent members for which the distance of
the channels are given by 
$2N^2-\Tr\Gamma_1^{\sf Z_{4I}}-\Tr\Gamma_3^{\sf Z_{4I}}= 
\frac34(2N^2-2\Tr\Gamma_2^{\sf Z_{4I}})$, 
see Eq.~\eqref{dist-4}. As it is not possible to represent such a group with
qubit and qutrit channels, we are going to analyze the set of channels acting on
$4$-dimensional systems. A
particular example is given by a group,
$G_\Gamma^{\sf Z_{4I}}=\{\Gamma_i^{\sf Z_{4I}}= V_i\otimes\overline{V}_i\}_{i=0}^3$,
formed according to the group up to a phase,
$\tilde{G}_V^{\sf Z_{4I}}=\{\1_4,\; \diag[1,1,i,-i],\; \diag[1,1,-1,-1], \;
\diag[1,1,-i,i]\}$.
Since both $i$ and $-i$ appear simultaneously in the spectrum of the generator
$V_1$, the elements of this group are linearly independent, although the
elements of the group up to a phase are not linearly independent.\\

Recall that a cyclic group of order $4$ has one non-trivial subgroup of order
$2$. In all corresponding polytopes of this group, there exist accessible maps
belonging to the line connecting the identity map to the channel corresponding
to the generator of this subgroup. In tetrahedrons shown in  Fig.~\ref{4-cyc}
this edge of the polytope contains accessible maps, while the diagonal of the
square in  Fig.~\ref{two-dim}(c) also contains accessible channels. Moreover,
let us emphasize that independently of the shape of the polytope representing a
cyclic group of order $4$,  the spectra of superoperators corresponding to
accessible channels form the same set in the complex plane as accessible maps in
the square of the quantum channel itself, see Fig. \ref{spectra} (b).

\begin{figure}[t]
\centering
\includegraphics[scale=0.4]{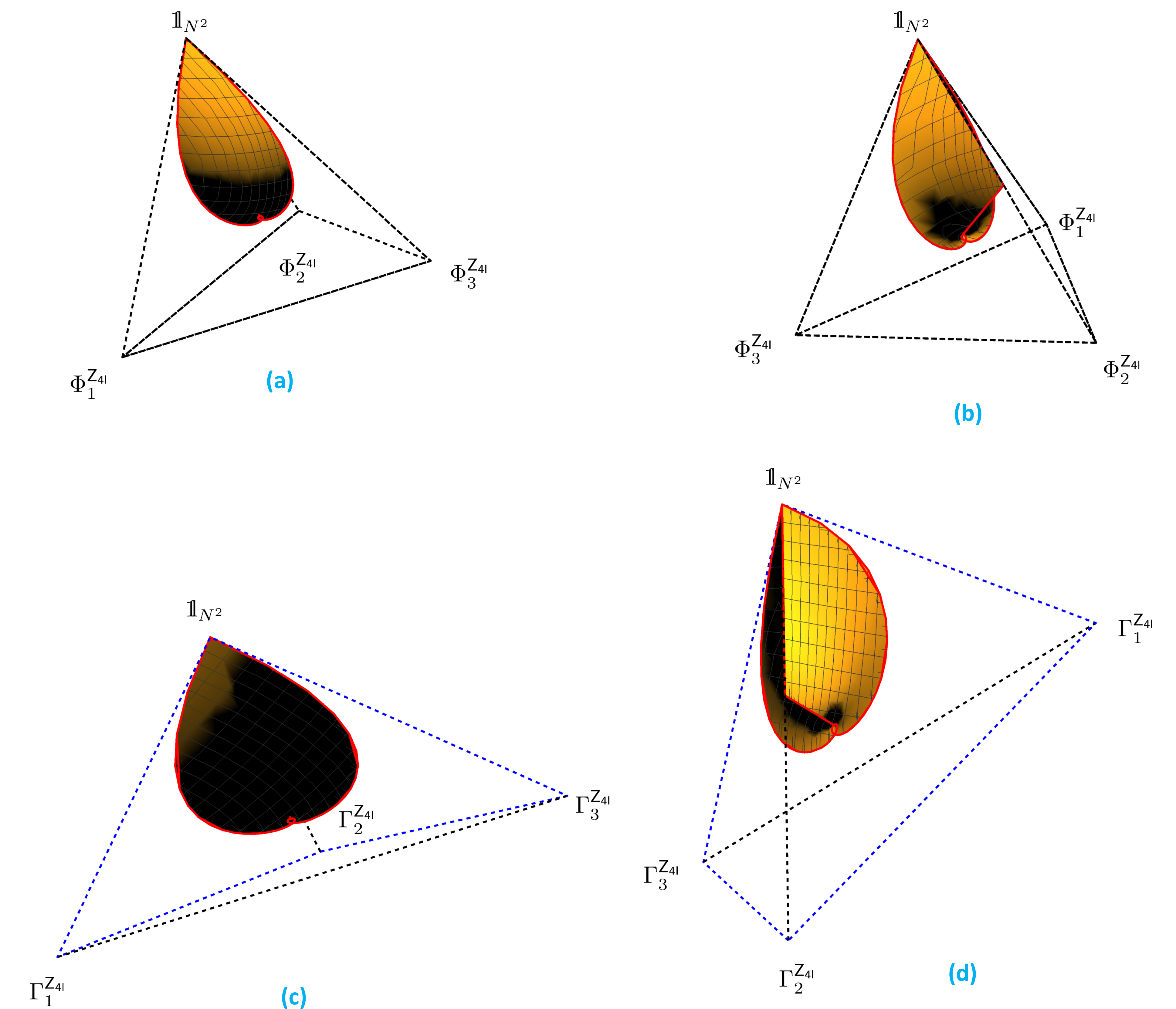}
\caption{
The subset of accessible quantum channels specified inside the 3-simplex: (a) and 
(b) related to a cyclic group
of order $4$ from two different perspectives, and
(c) and (d) inside the tetrahedron of a cyclic group of order $4$ whose elements 
are not equidistant.
The length of the dashed black edges of the tetrahedrons is $1$, and for the blue 
dashed ones, the length is equal to $\sqrt3/2$.
Explicit examples of quantum channels forming such groups are provided 
in Section \ref{cyclic}.
The volume of the subset of accessible maps
in both regular and non-regular shapes is $3(1-\e^{-4\pi})/32$ of its volume of the 
corresponding tetrahedron.
}
\label{4-cyc}
\end{figure}

\subsubsection{The non-cyclic group of order $4$}\label{ncycle}
As mentioned above the non-cyclic group
$G_\Psi=\{\Psi_0=\1,\Psi_1,\Psi_2,\Psi_3\}$ of order four is abelian. Due to the
fundamental theorem of abelian groups, it can be obtained by the product of two
cyclic groups of order $2$ whose elements are compatible with each other. Since
all non-trivial group members are of order $2$, their spectrum consists of
$\pm1$. There exist distinguished elements of the group if there are subspaces
in which the eigenvalues of all channels are not of the same sign. As for these
non-trivial elements, relation $\Psi_i\Psi_j=\Psi_k$ holds; the negative
eigenvalue should appear in the spectrum of two maps simultaneously in the same
subspace, so the third channel possesses a positive eigenvalue. Taking into
account all possibilities, one can  prove that to have $\sum c_i\Psi_i=0$ the
only possibility is $c_i=0$ for any $i$. This means the elements of $G_\Psi$ are
linearly independent. Thus the polytope for this group, regardless of the
dimension of the channels, is a tetrahedron, see Fig.~\ref{4-ncyc}. Moreover,
one needs to study the distance of the channels to find whether this tetrahedron
is regular or not. In this case, the Hilbert-Schmidt distance \eqref{HSdistance}
for different maps are given by
\begin{equation}
\begin{split}
D(\Psi_0,\Psi_i)&=2(N^2-\Tr\Psi_i), \quad\quad {\rm for}\ \ i\in \{1,2,3\},
\\
D(\Psi_i,\Psi_j)&=2(N^2-\Tr\Psi_k), \quad\quad {\rm for\ }\ i\neq j\neq k\ \ 
\rm{and}\ \ i,j,k\in \{1,2,3\}.
\end{split}
\end{equation}
As mentioned in Example \ref{g=2}, the number of negative eigenvalues in a
unitary channel of order $h=2$ is $2m(N-m)$ for an $m\in\{1,\dots,[N/2]\}$. In
particular, for qubit and qutrit channels, the numbers of negative eigenvalues
can be $2$ and $4$, respectively. This means unitary qubit channels of order $2$
are traceless and for any unitary qutrit channels of this order, the trace of
all such channels is equal to $5-4=1$. Thus for qubit and qutrit channels, the
tetrahedron spanned by the elements of the group is regular  -- see
Fig.~\ref{4-ncyc} (a) and (b).\\

The group of Pauli channels is the well-known example of a non-cyclic group of
order $4$ whose accessibility has been studied in \cite{DZP18,PRZ19}. As an
example of qutrit channels forming such a group, consider three different
quantum channels gained by different permutations of the unitary operator,
$U=\diag[1,1,-1]$. In higher dimensions, the tetrahedron can be non-regular, and
to emphasize this fact, the channels will be denoted as
$G_{\Gamma}=\{\Gamma_0=\1,\Gamma_1,\Gamma_2,\Gamma_3\}$, but the group does not
contain the superscript $\sf Z_{4I}$. An example of a regular tetrahedron
spanned by channels in dimension $N=4$ we can take identity and three unitary
operations obtained by three different permutations of $U=\diag[1,1,-1,-1]$
corresponding to three different quantum channels. To get a non-regular
tetrahedron consider a group up to a phase, $\tilde{G}_U=\{U_0,U_1,U_2,U_3\}$
where $U_0=\1_4$, $U_1=\diag[1,-1,-1,1]$, $U_2=\diag[1,1,1,-1]$, and
$U_3=\diag[-1,1,1,1]$. This tetrahedron is similar to the non-regular
tetrahedron of Example 3, in which the length of two non-connected edges are
$2/\sqrt{3}$ times the length of the other edges.
 \\

Note that the property of  accessibility of a channel is independent of the
dimension and the form of the polytope, as it depends on the group structure. In
the case of a non-cyclic group of order four all non-trivial elements are of
order $2$, so individual trajectories, with only a single non-zero interaction
time, are of the form  \eqref{c2}. For an accessible map, 
$\Psi_t=\sum w_i(t)\Psi_i$, with $t=t_1+t_2+t_3$ we get
\begin{equation}
	\begin{split}
		w_0(t)=&\frac14\left(1+\e^{-2(t_2+t_3)}+\e^{-2(t_1+t_3)}+
		\e^{-2(t_1+t_2)}\right),\\
		w_1(t)=&\frac14\left(1+\e^{-2(t_2+t_3)}-\e^{-2(t_1+t_3)}-
		\e^{-2(t_1+t_2)}\right),\\
		w_2(t)=&\frac14\left(1-\e^{-2(t_2+t_3)}+\e^{-2(t_1+t_3)}-
		\e^{-2(t_1+t_2)}\right),\\
		w_3(t)=&\frac14\left(1-\e^{-2(t_2+t_3)}-\e^{-2(t_1+t_3)}+
		\e^{-2(t_1+t_2)}\right).
	\end{split}
\end{equation}
These results have already been reported for the Pauli channels
\cite{DZP18,PRZ19,SAPZ20}. However, here we see that they can be gained for any
set of quantum channels admitting the same group structure. The subset of
accessible maps in the regular and non-regular tetrahedrons are presented in
Fig.~\ref{4-ncyc}. Note that a non-cyclic group of order $4$ has three
non-trivial subgroups, each of which is of order $2$. Consequently, we see on
the edges connecting the identity map to other three elements, accessible maps
can be found.

\begin{figure}
\centering
\includegraphics[scale=0.5]{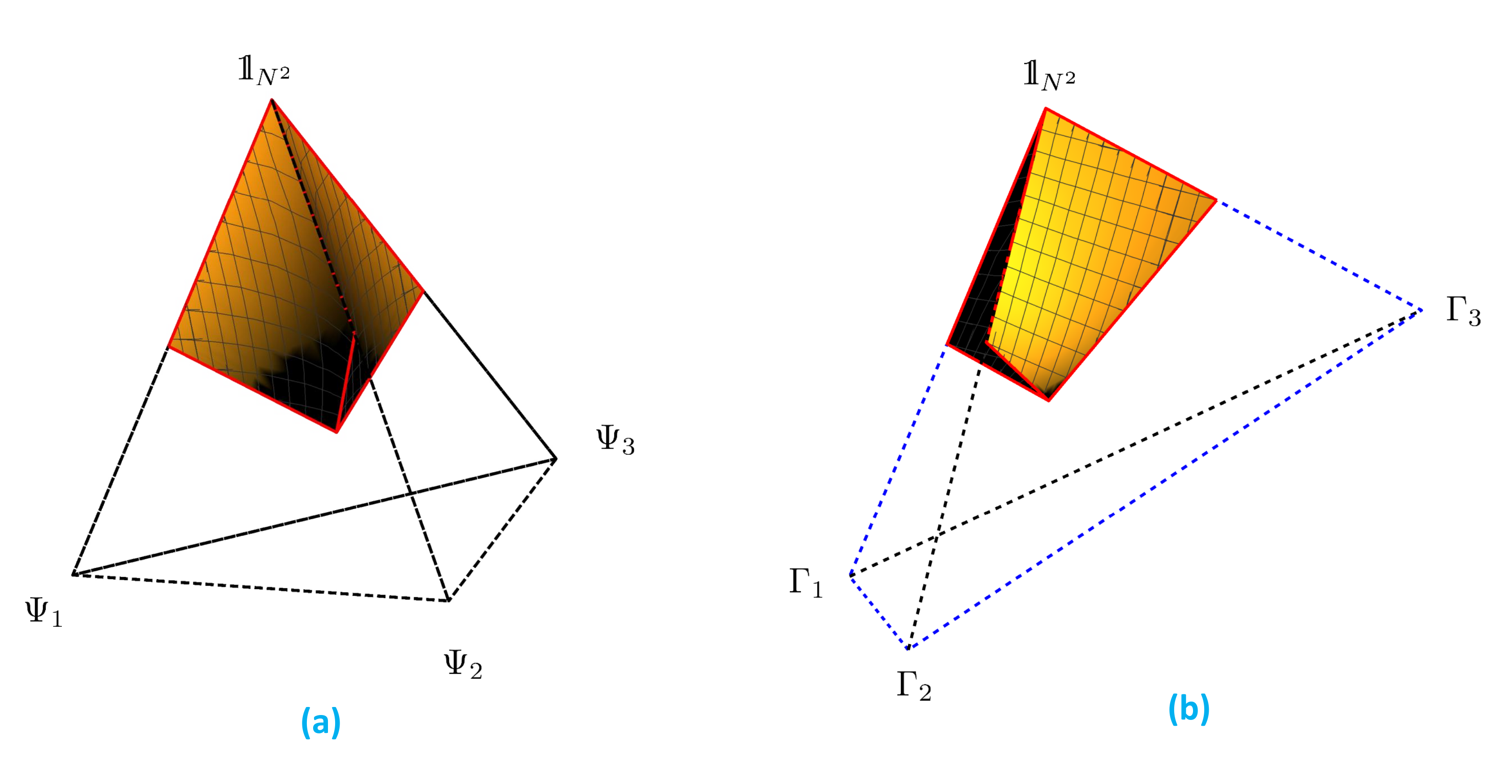}
\caption{
The set  $\cal A$ of accessible maps specified in the space of quantum channels
forming a non-cyclic group of order $4$ for (a) a 3-simplex, regular
tetrahedron which includes the case of Pauli channels
\cite{DZP18,PRZ19,SAPZ20}), and (b) a non-regular tetrahedron. Black dashed
edges are of length $1$, while blue dashed edges have the length $\sqrt3/2$.
Explicit examples of these groups are discussed in Section \ref{ncycle}.
Accessible maps in each of these polytopes occupy $3/32$ of the volume of their
corresponding tetrahedrons.
}
\label{4-ncyc}
\end{figure}
\end{example}

\section{Geometric properties of the set of accessible maps}\label{geometry}

In this section, we will study the geometric properties of the set ${\cal A}$ of
accessible maps, which are obtained by a mixture of unitary maps forming a
discrete group. Proposition \ref{limit} implies that any Lindblad operator
determined by a quantum channel, $\c L_\Phi=\Phi-\1$, generates a trajectory in
the set of quantum channels which converges to the projector onto the invariant
subspace of the channel $\Phi$. For any quantum channel $\Phi$, the line
connecting the identity with  $\Phi$ is  a tangent line to the trajectory
generated by $\c L_\Phi$ at the beginning of the evolution. To see this, note
that for the dynamical semigroup $\exp(t\c L_\Phi)$, one gets
\begin{equation} 
\dfrac{{\rm d}}{{\rm d} t}\ \e^{(t\c L_\Phi)}\Big|_{t=0^+}=
\c L_\Phi \ \e^{(t\c L_\Phi)}\Big|_{t=0^+}=\Phi-\1,
\end{equation}
which proves the claim that $(1-p)\1+p\Phi$ defines the tangent 
line for the evolution at $t=0$.
Analyzing a finite group of quantum channels and the semigroup formed by convex
combinations of elements of this group, a distinguished trajectory joins the
identity with the centre $\Phi_*$ of the polytope, which follows the direction
determined by the tangent line at the beginning of evolution. Stated in other
words, we have the following result.

\begin{proposition}\label{central-line}
	Let $G_\Phi=\{\Phi_\mu\}$ define a finite group of $g$ quantum channels.  The
	line connecting identity to the centre of the polytope,
	$\Phi_*=\frac{1}{g}\sum_{\mu=0}^{g-1}\Phi_\mu$, always belongs to the set  $\cal
	A$ of accessible channels in this group.
\end{proposition}

\begin{proof}
To show this we will start with the Lindbladian 
$\c L_*=\frac{1}{g-1}\sum_{\mu=1}^{g-1}\Phi_{\mu}-\1$.
Note that in this definition we ignored the trivial generator $\c L_{\Phi_0}=0$
to follow the convention adopted here. 
So $\c L_*$ is not equal to $\Phi_*-\1$, but $\c L_*=\frac{g}{g-1}(\Phi_*-\1)$.
However, they generate the same dynamics up to a rescaling of time.
\begin{equation}\label{Lstar}
\begin{split}
	\varphi(t)&:=\e^{t\c L_*}=\e^{\frac{-tg}{g-1}}
	\e^{\frac{tg}{g-1}\Phi_*}=
	\e^{\frac{-tg}{g-1}}\left(\1+\sum_{m=1}^\infty
	\frac{(\frac{tg}{g-1})^m}{m!}\Phi_*^m\right)
	\\&=
	\e^{\frac{-tg}{g-1}}\left(\1+\sum_{m=1}^\infty
	\frac{(\frac{tg}{g-1})^m}{m!}\Phi_*\right)=\e^{\frac{-tg}{g-1}}\1+
	(1-\e^{\frac{-tg}{g-1}})\Phi_*,
\end{split}
\end{equation}
where we called the resultant channel $\varphi(t)$ for future reference.
Furthermore, in the first equation of the second line, we used the fact that 
$\Phi_*$ is a projector; see Proposition \ref{infinty}. Eq. \eqref{Lstar} shows the 
line connecting the identity to the centre of polytope always belongs to the set of 
accessible maps and completes the proof.
\end{proof}


The above proposition leads us to the following important result.
\begin{theorem}\label{nonzerovolume}
	The set $\cal A$ of accessible maps occupies a positive measure in the polytope 
	formed by the convex hull of the corresponding group elements.
\end{theorem}
\begin{proof}
	To prove, we will show that for any non-ending point in the interval connecting 
	identity to the centre of the polytope, one can always find a ball that belongs 
	to the set of accessible maps.
	According to above proposition, any such non-ending point can be defined as 
	$\exp[(T+\delta T)\c L_*]$ for time $t=T$ such that there exists $\delta T$ 
	with $|\delta T|<<T$. 
	Now let $\forall\mu\in{1,\cdots,g-1}$ we have $|\epsilon_\mu|<<1$  and  
	$\sum_{\mu=1}^{g-1}\epsilon_\mu=0$.
	So we can define a family of valid $\vec{\epsilon}$-dependent Lindbladians as 
	$\c L_{\vec{\epsilon}}=\sum_{\mu=1}^{g-1} 
	(\frac{1}{g-1}+\epsilon_\mu)\Phi_\mu-\1$. Hence, for any $\vec{\epsilon}$ 
	satisfying above conditions $\exp(t\c L_{\vec{\epsilon}})$ belongs to the set 
	of accessible maps at each moment of time. Thus, for $t=T+\delta T$ we get
\begin{equation}\label{ball}
	\begin{split}
	\e^{(T+\delta T)\c L_{\vec{\epsilon}}}&=
	\e^{(T+\delta T)(\c L_*+\sum\epsilon_\mu\Phi_\mu)}=
	\e^{(T+\delta T)\c L_*}\e^{(T+\delta T)\sum\epsilon_\mu\Phi_\mu}
	\\&\approx
	\varphi({T+\delta T})\left(\1+(T+\delta T)
	\sum_{\mu=1}^{g-1}\epsilon_\mu\Phi_\mu\right)
	\\&\approx
	\varphi(T)+\e^{-\frac{Tg}{g-1}}\sum_{\mu=0}^{g-1}\epsilon'_\mu\Phi_\mu,
	\end{split}
\end{equation}
	where $\epsilon'_0=-\delta T$ and $\epsilon'_\mu=T\epsilon_\mu+\delta T/(g-1)$ 
	for $\mu\in\{1,\cdots,g-1\}$.
	In writing these equations, we considered only the first order of $\delta T$ 
	and $\vec{\epsilon}$ and also ignored their products. Moreover, one may notice 
	that $\Phi_*$ and consequently $\c L_*$ commute with any channel in the 
	polytope of the group as stated in Proposition \ref{infinty}. The left-hand 
	side of Eq. \eqref{ball} belongs to the set of accessible maps of the group 
	$G_\Phi$, so does a neighbourhood of any $\varphi(T)$ for a finite $T$ in any 
	direction. 
	This completes the proof.
\end{proof}

The above result was intuitive because the trajectory generated by Lindbladian 
corresponding to any quantum channel starts in the tangent plane of the line 
connecting the identity element to the assumed channel. 
So the set of accessible maps inside a polytope is defined by trajectories that go 
in all possible directions at the beginning. However, we will generalize the above 
approach concerning the point $\varphi(T)$ in the sequel to show for an abelian 
group at a finite time for any point on the trajectory generated by a Lindbladian 
corresponding to a channel from the interior of the polytope such a ball exists. 
Before proceeding with such an extension, let us conclude the discussion on the 
volume by explicit calculation of the volume for the examples mentioned in the 
previous section.
\begin{proposition} \label{volumes}
The ratio of the volume of the set $\cal A$  of accessible maps to the volume 
of  the corresponding polytope for the examples discussed in 
Section \ref{Examples} reads:
\begin{enumerate}[a)]
\item  For any choice of a cyclic group of order $g$ with exactly three 
linearly independent elements corresponding to a regular $g$-polygon
in the plane, see Fig.~\ref{two-dim}, the ratio is given by 
\begin{equation}\label{polygon volume}
	\frac{V_{acc}}{V_{G^{\sf Z_{gD}}}}=
	\frac{1-\e^{-2\pi\tan(\pi/g)}}{2g\sin^2(\pi/g)}
	= 1-\frac{\pi^2}{g}+\c O(\frac{1}{g^2}).
\end{equation}
\item For a cyclic group of order $g=4$ with all linearly independent
elements, for both regular and non-regular tetrahedrons, 
see Fig.~\ref{4-cyc},
the ratio is given by
\begin{equation}\label{corrected volume}
	\frac{V_{acc}}{V_{G^{\sf Z_{4I}}}}=\frac{3}{32}(1-\e^{-4\pi}) \approx
	  0.0937497 .
\end{equation}
\item For the non-cyclic group of order $g=4$
the ratio is given by
\begin{equation}\label{pauli volume}
	\frac{V_{acc}}{V_{G_\Psi}}=\frac{3}{32} = 0.09375,
\end{equation}
independently of the regularity of tetrahedron, see Fig.~\ref{4-ncyc}.
\end{enumerate}

\end{proposition}

The proof of this proposition is provided in Appendix \ref{AppA}. Note that 
Eqs.~\eqref{polygon volume} and \eqref{pauli volume} were already presented in 
\cite{SAPZ20} for the ratio of the spectra of accessible to the entire set of Weyl 
maps and for accessible Pauli channels, respectively. However, we should emphasize 
that these results are valid for any representation of the aforementioned groups 
independently of the dimension of the maps. In Remark \ref{weyl measure} we show 
that for the specific example of Weyl channels, the volume of the set of accessible 
and Markovian Weyl channels are the same. Accordingly, in this case, the set $\cal 
A$ of the accessible maps forms a good approximation of the set of Markovian 
channels. Now we will apply the same approach of Theorem \ref{nonzerovolume} to get 
the boundaries of the  set $\cal A$ assigned to an 
abelian group.

\begin{theorem}\label{boundary}
The boundary $\partial {\cal A}$ of the set of accessible channels in the
polytope of a finite abelian group of unitary maps is formed by Lindblad
operators $\c L_\Phi =\Phi-\1$, where $\Phi$ belongs to the boundaries of the
polytope.
\end{theorem}

\begin{proof}
Let $\c L=\sum q_\mu\Phi_\mu-\1$ denote the Lindblad generator such that the
channels $\Phi_\mu$  along with identity form a group of order $g$ and
$q_\mu\neq0$ for all $\mu$. Then it is always possible to find a probability
$q'_\mu=q_\mu+\epsilon_\mu$ for small enough $\epsilon_\mu$ such that
$\sum\epsilon_\mu=0$. Similar to the proof of the previous theorem, we can
define a family of valid $\vec{\epsilon}$-dependent Lindblad generators $\c
L_{\vec{\epsilon}}= \sum_{\mu=1}^{g-1}(q_\mu+\epsilon_\mu)\Phi_\mu-\1$. So for
any finite time $t=T$ there is $\delta T$ such that $|\delta T|<<1$ and
$\exp[(T+\delta T)\c L_{\vec{\epsilon}}]$ belongs to the set of accessible maps.
On the other hand, we have
\begin{equation}\label{abelianball}
	\begin{split}
	\e^{(T+\delta T)\c L_{\vec{\epsilon}}}&=
	\e^{(T+\delta T)\c L}\e^{(T+\delta T)\sum_{\mu}\epsilon_\mu\Phi_\mu}
	\\&\approx
	\e^{T\c L}\left((1-\delta T)\1+
	\delta T\sum_{\mu=1}^{g-1}q_\mu\Phi_\mu\right)
	\left(\1+T\sum_{\mu=1}^{g-1}\epsilon_\mu\Phi_\mu\right)
	\\&\approx
	\e^{T\c L}+\e^{T\c L}\left(-\delta T\1+
	\sum_{\mu=1}^{g-1}(T\epsilon_\mu+q_\mu\delta T)\Phi_\mu\right)
	\\&\approx
	\e^{T\c L}+\sum_{\alpha=0}^{g-1}\sum_{\mu=0}^{g-1}w_\alpha(T)
	\epsilon'_\mu\Phi_\alpha\Phi_\mu=\e^{T\c L}+
	\sum_{i=0}^{g-1}f_i\left(\vec{w}(T),\vec{\epsilon}'\right)\Phi_i,
	\end{split}
\end{equation}
where $\epsilon'_0=-\delta T$, $\epsilon'_\mu= T\epsilon_\mu+q_\mu\delta T$ for
$\mu\in\{1,\cdots,g-1\}$, and $w_\alpha(T)$ is the time-dependent probability by
which the trajectory of $\exp(T\c L)$ is defined, see Theorem \ref{group}. Note
that due to rearrangement theorem for any fixed $\alpha$ the term
$\Phi_\alpha\Phi_\mu$ produces all group elements with different order according
to the group multiplication table. Thus we can define the $g$-dimensional vector
$\vec{f}\left(\vec{w}(T),\vec{\epsilon}\right)$ by its element, 
$f_i\left(\vec{w}(T),\vec{\epsilon}\right)$ which can be gained by 
\begin{equation}
	\vec{f}\left(\vec{w}(T),\vec{\epsilon}'\right)=(\sum_{\alpha=0}^{g-1}
	w_\alpha(T)R_\alpha)\vec{\epsilon}';
\end{equation}
where $R_\alpha$ is the  regular representation of element $\alpha$ of the group
$G=\{\Phi_0=\1,\cdots,\Phi_{g-1}\}$ of dimension $g\times g$ defined based on
the group multiplication table see Lemma \ref{regular}. Therefore, one can claim
the vector $\vec{f}\left(\vec{w}(T),\vec{\epsilon}'\right)$ can define a ball
around the point $\exp(T\c L)$ if the $g$-dimensional matrix
$\sum_{\alpha=0}^{g-1}w_\alpha(T)R_\alpha$ is invertible. To show that  we use
Lemma~\ref{regular} that leads to
\begin{equation}\label{multi-table}
	\sum_{\alpha=0}^{g-1}w_\alpha(T)R_\alpha=
	\exp\left[T\left(\sum_{\mu=1}^{g-1}q_\mu R_\mu-\1\right)\right].
\end{equation}
This is an invertible matrix as long as $T$ is finite, which is the case here, 
completing the proof.
\end{proof}

We should notice that the inverse of the theorem  is not 
necessarily valid, i.e. not any point from the boundaries of a polytope
formed by a finite abelian group of quantum channels can give a 
trajectory at the boundaries of the set of accessible maps.
For a counterexample, see all the red dashed lines in 
Fig.~\ref{two-dim}c-f. These are trajectories generated by 
single Lindbladians corresponding to the channels at different 
vertices of the polygons.
Our conjecture is that the inverse is true when the elements are
linearly-independent.\\

To proceed with investigating more geometrical properties of accessible maps, in 
what follows, we show that this set is star-shaped with respect to the centre of 
its corresponding polytope.
First, note that due to Proposition \ref{central-line},
the line segment $[\1 , \Phi_* ]$ is accessible
by the generator $\c L_*$ that commutes with all channels in the polytope. Applying 
this fact, we get that for any accessible $\Omega_t$
\begin{eqnarray}
	p \Omega_t + (1-p) \Phi_*=
	\Omega_t (p \1 + (1-p) \Phi_*)= 
	\e^{t\c L}\e^{t_\ast\c L_\ast}=\e^{t'\c L'},
\end{eqnarray}
and therefore, the line segment $[\Omega_t, \Phi_*]$ is accessible, 
proving the following result.

\begin{theorem}\label{star}
The set $\cal A$ of accessible maps in a polytope formed by a discrete finite group 
of quantum channels is star-shaped with respect to the centre of the polytope.
\end{theorem}

\begin{corollary}\label{star2}
If a set  $\cal A$ of accessible maps is planar, then accessible maps 
are star-shaped with respect to the whole interval $[\1, \Phi_*]$. 
\end{corollary}

\begin{proof}
	This follows from the fact, that for any accessible 
	$\Omega_T = e^{T\mathcal{L}}$, the set
	\begin{equation}
		\{p e^{t \mathcal{L}} + (1-p) \Phi_* \}_{p\in[0,1],\ t \in [0,T]}
	\end{equation}
	is accessible. Next using the fact, that this is a subset of a plane, we have 
	for all $q \in [0,1]$, it contains intervals of a form
	\begin{equation}
		\{p e^{T\mathcal{L}} + (1- p) (q \1 + (1-q) \Phi_*) 
		\}_{p\in[0,1]}.
	\end{equation}
	Which shows the star shapeness with respect to the interval $[\1 , \Phi_*]$.
\end{proof}

\medskip


\section{Accessibility rank, Pauli and Weyl channels}\label{last examples}

In this Section, we explicitly study some further examples of quantum channels forming a 
group and discuss the rank of accessible channels.

\begin{example}\normalfont
Let $\lambda_j=\e^{i\theta_j}$ with $j=1,\dots,N^2$
denote eigenvalues of a superoperator
$\Phi_U=U\otimes\overline{U}$
corresponding to a unitary $U$.
If $\forall j:\ \theta_j=2\pi r_j$, where
$r_j=m_j/n_j$ is a irreducible fraction of a  rational number, then $\Phi$ can 
generate a discrete and finite cyclic group of order equal to the lowest common 
multiple of $n_j$, denoted by $n$, i.e. $G_\Phi=\{\Phi^i\}_{i=0}^{n-1}$.

As an example of such a group, one may consider a group of rotations around a
fixed axis with discrete angles. Let us focus on qubit channels. In this case a
rotation around $\hat{n}$ of magnitude $\theta$ is gained by applying
$U=\e^{i\frac\theta2(\hat{n}.\vec{\sigma})}$. Let us take $\theta$ to be one of
the following angles $\{\theta_j=\frac{2\pi j}{l}\}$ with $j=0,1,\dots, l-1$. It
is easy to check that
$\{U_j=\exp\left(i\frac{\theta_j}{2}(\hat{n}.\vec{\sigma})\right)\}$ is an
abelian group up to a phase of order $l$. Thus $G_\Phi=\{\Phi_j=U_j\otimes
\overline{U}_j\}_{j=0}^{l-1}$ is an abelian (more precisely a cyclic) group of
the same order. Moreover, for every divisor $d$ of $l$ ($l=dd'$), $G_\Phi$ has
at most one cyclic subgroup $H_\Phi$ of order $d$ generated by different powers
of $\Phi_{d'}=U_{d'}\otimes\overline{U}_{d'}$. Therefore, one can always find an
accessible map of order $h_\Phi=d$ which has a convex expansion in terms of the
subgroup members $H_\Phi$. The space of such a group is an $l$-polygon. For the
$l=2,\dots,7$, the set of channels and their corresponding accessible maps are
given in Fig.~\ref{two-dim}.

\end{example}

\begin{example}[\bf Group of commutative quantum maps forming an orthogonal basis]\normalfont

 An abelian group of $N^2$ unitary quantum maps forming an orthogonal basis,
$\tr(U_\mu U^\dagger_\nu)=N\delta_{\mu\nu}$, is one of the simplest non-trivial
examples to investigate. Let $\tilde{G}_U=\{U_\mu\}_{\mu=0}^{N^2-1}$ denote an
abelian group up to a phase of such $N^2$ unitary matrices. Weyl unitary
matrices \cite{W27} provide an example of such a set. Since we are dealing with
an abelian group here, Corollary~\ref{abelian accessibility} clarifies the
problem of accessibility. Since unitary matrices form an orthogonal basis in
the Hilbert-Schmidt space of matrices and thus are linearly independent, the
analyzed polytope forms a simplex in dimension $N^2-1$. This implies that the
convex combinations we get through Corollary~\ref{2cyclic lindblad} or
Corollary~\ref{abelian accessibility} for the accessible maps are unique.
Therefore, we can call the number of vertices appearing in such a unique convex
combination the rank of the channel (and it equals to the Choi rank as well
since the vertices are orthogonal maps). Accordingly, an immediate result of
Corollary~\ref{abelian accessibility} is that not all quantum maps with any
assumed rank can be accessible.

\begin{proposition}
	The set  $\mathcal{A}_N^Q$ of accessible quantum maps of the form $\Phi=\sum
	p_\mu\Phi_\mu$, in which $\Phi_\mu$'s are unitary channels belong to an abelian
	group and satisfy
	$\Tr\left(\Phi_\mu\Phi_\nu^\dagger\right)=N^2\delta_{\mu\nu}$, is formed by the
	maps of  (Choi) rank equal to the order of possible cyclic subgroups of
	$G_\Phi$ or their multiplications.
\end{proposition}

Through the fundamental theorem of abelian groups this the theorem can be
rephrased as follows. The set of accessible channels satisfying the above
conditions is formed by the maps of Choi rank equal to the order of possible
subgroups of $G_\Phi$. Note that through Corollary~\ref{abelian accessibility}
it is possible not only to indicate the rank of accessible maps but also to find
which $\Phi_\mu$ participated in the combination when an assumed $t_\nu$ is not
zero. In summary, we should expand convexly an assumed quantum channel $\Phi$ in
terms of extreme points $\Phi_\mu$. Such an expansion is unique here. If the
extreme points participating in this expansion do not form a subgroup of
$G_\Phi$,  then the map $\Phi$  is not accessible.\\

As an example let us consider Weyl channels more explicitly.
Weyl channels are a unitary generalization of Pauli maps. 
These channels are based on $N^2$ unitary operators of the form $U_{kl}=X^kZ^l$ for 
$k,l\in\{0,\dots,N-1\}$ where $X\ket i=\ket{i\oplus1}$ and 
$Z=\diag[1,\omega_N,\dots,\omega_N^{N-1}]$ with 
$\omega_N=\exp(i\frac{2\pi}{N})$. 
The accessibility in this set is already studied in~\cite{SAPZ20}.
In the case $N=2$ of Pauli channels $G_\Phi$ is of order $4$.
There are three different cyclic subgroups of order $2$,
$H_i=\{\1,\sigma_i\otimes\overline{\sigma}_i\}$ for $i\in\{1,2,3\}$ multiplication 
of any two of them exhausts the group $G_\Phi$.
Moreover, the identity element itself is always the only trivial cyclic subgroup of 
order  $1$. Therefore, there are: 
(i) The identity map as the only accessible map of rank $1$.
(ii) Accessible maps of rank two in the form of 
$\Phi=p\1+(1-p)\sigma_i\otimes\overline{\sigma}_i$ for $i\in\{1,2,3\}$ which are 
gained when the only non-vanishing interaction time is $t_i$.
(iii) The full-rank accessible maps, which are available when there are at least 
two different $t_i\neq0$. Fig.~\ref{4-ncyc}(a) presents the accessible
maps among all Pauli channels.\\

In the case of qutrit Weyl maps, $G_\Phi$ is of order $9$. There are four
different cyclic subgroups of order $3$, $H_{\Phi}=\{\1,\Phi_{kl},\Phi_{-k\oplus
	N,l\oplus N}\}$. Again, the multiplication of any two subgroups makes the group
$G_\Phi$. In addition, there is not any other nontrivial subgroup for $G_\Phi$.
Therefore, the accessible maps belong to one of the following sets: (i) The
identity map as the only rank one accessible map. (ii)Accessible maps of rank
$3$ which can be found on the face of the simplex specified by
$(\1,\Phi_{kl},\Phi_{-k\oplus N,-l\oplus N})$ when $t_{kl}\neq0$ or $t_{-k\oplus
	N,-l\oplus N}\neq 0$ and all other interaction times are zero.
Fig.~\ref{two-dim}(b) can be considered as a cross section presenting such a
subgroup. (iii)The full-rank accessible maps which are available when there are
at least two different $t_{kl}\neq0$ and $t_{k'l'}\neq 0$ where $(k',l')\neq
(-k\oplus N,-l\oplus N)$.\\

The above results concerning the Weyl channels are mentioned in \cite{SAPZ20}.
However, the group properties of Weyl channels help us to generalize these
results to a general case of $N$-dimension. Any group of Weyl unitary maps
acting on $N$-dimensional systems, formed by $N+1$ cyclic subgroups of order $N$
with no element in common but the identity. This implies that the direct product
of any two of them forms the group entirely. So we get
\begin{corollary}
	In the case of Weyl quantum maps acting on $N$-level systems, there always exist
	accessible maps of rank $1$, $N$ and $N^2$. When $N$ is prime, these are the
	only choices, however, for a composite $N$ other ranks are possible too. The
	ranks correspond to the cardinality of the subgroups.
\end{corollary}

Before proceeding with the following example concerning local channels, let us
mention the following simple and yet significant remark about accessible Weyl
channels.
\begin{remark}\label{weyl measure}
The accessible Weyl channels can recover the full measure of Markovian Weyl 
channels but not the entire set.
\end{remark}

The supporting evidence for this statement comes from the fact that the
necessary and sufficient condition for accessibility of Weyl channels introduced
in \cite{SAPZ20} is the same as the necessary and sufficient condition stated in
\cite{WECC08} for Markovianity of channels with non-negative eigenvalues. Note
that in the case of channels with negative eigenvalue, Markovianity can only be
defined for the maps whose negative eigenvalues are even-fold degenerated
\cite{DZP18}. Such a set is by definition of measure zero.
\end{example}

\begin{example}[\bf Tensor product of quantum channels]\normalfont

Let us now generalize our results to channels acting locally in higher dimension.
Note that the set $\tilde{G}_{U\otimes V}=\{U_\mu\otimes V_\alpha\}$ is an
abelian group up to a phase if $\tilde{G}_U=\{U_\mu\}$ and
$\tilde{G}_V=\{V_\alpha\}$ are.
Moreover, it contains a complete set of orthogonal
matrices, i.e. $\Tr\left((U_\mu\otimes V_\alpha)
(U_\nu\otimes V_\beta)^\dagger\right)=
\Tr(U_\mu U^\dagger_\nu)\Tr(V_\alpha V^\dagger_\beta)=
MN\delta_{\mu\nu}\delta_{\alpha\beta}$, provided that $\tilde{G}_U$ and
$\tilde{G}_V$ are two sets of orthogonal basis in the space of matrices of order $M$ and $N$, respectively.
Let $G_\Upsilon$ denote the group formed by  $\tilde{G}_{U\otimes V}$.
We will also assume that $\tilde{G}_U$ and $\tilde{G}_V$ are the sets
of orthogonal basis and are abelian up to a phase.
Therefore,  the subgroups of $G_\Upsilon$ are of 
order $h_\Phi h_\Psi$ where $h_\Phi$ and $h_\Psi$ are 
the cardinalities of the subgroups of 
$G_\Phi=\{\Phi_\alpha=U_\alpha\otimes\overline{U}_\alpha\}$ and 
$G_\Psi=\{\Psi_\mu=V_\mu\otimes\overline{V}_\mu\}$, respectively,
and each of which is a group of orthogonal quantum maps.\\


In a very particular case, when we restrict ourselves to $G_\Phi=G_\Psi$ being
the group related to Pauli channels, we get $\{1,2,4\}^{\otimes
	2}=\{1,2,4,8,16\}$ as the rank of an accessible map of the form $\c E_A\otimes\c
E_B$ acting on a two-qubit system. It is easy to generalize this result to the
system of $k$ qubits, each of them is subjected to an evolution described by a
Pauli map.

\begin{corollary}
Let $ G^{\otimes k}$ define the group obtained by a tensor product of Pauli 
channels. Hence each element is a map acting on a $k$-qubit system. Then channels 
accessible by Pauli semigroups can only be of rank equal to $2^m$ with 
$m=0,\dots,2k$.
\end{corollary}

\end{example}

\section{ Classical stochastic maps} \label{classical}

Any stochastic matrix $T$ describes a Markov chain -- a transition in space of
$N$-point probability distributions, $p'=Tp$. By construction, each colum on $T$
forms an $N$-point probability vector itself. Return now back to the space of
density matrices. The process of decoherence transforms a quantum state $\rho$
into a classical probability vector embedded in its diagonal entries,
$\rho_d=\Delta(\rho)=\sum\rho_{ii}\project{i}$, where $\Delta$ denotes the
decoherence channel. Similarly, a transition matrix $T$ can be obtained from a
quantum channel suffering the super-decoherence \cite{KCPZ18},
\begin{equation}
T_{ij}(\c E)=\bra i\c E(\project j)\ket i.
\end{equation}
This effect can be considered as a decoherence applied to the Choi matrix
$C=d(\c E\otimes {\mathbbm I})P_+$, where $P_+=|\psi_+\rangle \langle \psi_+|$
with $|\psi_+\rangle= \sum_i |i,i\rangle/\sqrt{N}$ denotes the maximally
entangled state. Hence the classical transition matrix  $T$ arises from
reshaping of the $N^2$ diagonal entries of  $C$ into a square matrix of size
$N$. Super decoherence of a quantum channel $\c E$ can  also be described by
sandwiching it between repeated action of the decoherence channel,
\begin{equation}\label{equiv}
T(\c E)\equiv\Delta\circ\c E\circ\Delta.
\end{equation}
If $\{K_a\}$ denotes the set of  Kraus operators of the channel $\c E$, then
$T(\c E)$ is given by Hadamard product, $T(\c E)=\sum_a K_a\odot\overline{K}_a$.
Thus super-decoherence transforms  a unitary quantum channel, $\Phi_U=U \otimes
{\bar U}$, into a unistochastic transition matrix, $T_U=U \odot {\bar U}$.
Moreover, one can apply the special ordering of Gell-Mann basis where
$B_0=\frac{\1_N}{\sqrt{N}}$, $B_i$'s for $i=1,\dots,N-1$ are the diagonal
elements, and for $i=N,\dots,N^2-1$ denote other elements. In that case, we can
divide the distortion matrix $M$ of size $N^2-1$ defined in \eqref{affine} into
four blocks and the translation vector into two parts,
\begin{equation}\label{affine-dec}
\Phi=\begin{pmatrix}
1&&0&&0&&\\
\vect{k}&&D&&Q\\
\vect{k'}&&Q'&&D'
\end{pmatrix},
\end{equation}
where $D$ and $D'$ are square matrices of order $N-1$ and $N^2-N$, respectively,
and $\vect{k}$, $\vect{k'}$ are vectors with the same respective length.
Off-diagonal blocks $Q$ and $Q'$ are rectangular so that the dimension of $\Phi$
is $N^2$. Thus, Eq. \eqref{equiv} implies that the process of super-decoherence
in the above basis results in a projection into a $N$-dimensional space
\cite{Ketal20},
\begin{equation}\label{affine,T}
T(\c E)=
\begin{pmatrix}
	1&&0\\
	\vect{k}&&D
\end{pmatrix}.
\end{equation}\\

It is worth mentioning that super-decoherence sends Lindblad generators  to the
set of Kolmogorov operators, i.e. the generators of classical Markovian
evolutions.  This fact can be observed through the GKLS form \eqref{gkls}. The
Kolmogorov generators $\c K$, also called `transition rate matrices',
satisfy conditions:
(a) $\c K_{ij}\geq0$ for all $i\neq j$, and
(b) $\sum_i\c K_{ij}=0$ for all $j$.

\begin{lemma}\label{inc}
	Let $\Phi$ \eqref{affine-dec} be a quantum channel and $T(\c E)$
	\eqref{affine,T} denote its corresponding classical transition gained by
	super-decoherence, then $T^n(\c E)=T(\c E^n)$ for all $n\in\{1,\dots,r\}$ if and
	only if (i) $Q(D')^{n-2}\vect{k'}=0$, and (ii) $Q(D')^{n-2}Q'=0$ for all
	$n\in\{2,\dots,r\}$.
\end{lemma}

\begin{proof}
The proof of this lemma follows from Eq.~\eqref{equiv} and the fact that in the 
Bloch representation \eqref{affine-dec} for the decoherence map $\Delta$ we have 
$D_\Delta=\1_{N-1}$, while other parameters vanish.
\end{proof}


If the above condition for a quantum channel holds, then the cyclic group
generated by different powers of this channel has its classical counterpart of
the same order, such that each element finds its respective element in the
classical set through super-decoherence.\\

A trivial example of quantum channels satisfying Lemma \ref{inc}
are those for which $Q=Q'=0$ and $\vect{k'}=0$. This implies
\begin{equation}\label{mio}
\Delta\circ\c E\circ\Delta=\c E\circ\Delta.
\end{equation}
The set of quantum channels satisfying above equation are known as Maximally
Incoherent Operation (MIO) in the context of resource theory of quantum
coherence \cite{BCP14}. This is the largest set of operations which are not able
to generate coherence in an incoherent state. So taking classical probabilities
embedded in diagonal entries of an incoherent state $\rho_d$, the set MIO can be
seen as classical operations.  Note that the super-decoherence map from quantum
channels to stochastic matrices is a surjective and non-injective function.
However, we can always find a maximally incoherent operation  which is sent to
any assumed classical transition through super-decoherence. Moreover, it is
worth mentioning if $\c E$ satisfies Eq. \eqref{mio}, then the corresponding
Lindblad generator also satisfies the same condition. Such a generator can be
thus called an {\sl incoherent generator}. Interestingly, for any Kolmogorov
generator $\c K$ one can always find an incoherent Lindblad generator $\c L$
\cite{KL21} whose trajectory  is an incoherent quantum channel at each moment of
time.\\

In the classical case, $G_T=\{T_i\}$ defines a group if $T_i$ are permutation
matrices, i.e. the extreme points of Birkhoff polytope. So the group $G_T$ is of
order $N !$ and an arbitrary convex combination is the most general bistochastic
matrix. We assign to each element of $G_T$ the Kolmogorov generator $\c
K_{T_i}=T_i-\1_N$. As any classical transition matrix is equivalent to an MIO,
and its Kolmogorov operator is equivalent to an incoherent Lindblad generator,
all the results of Section \ref{general} are also valid for classical
transitions. However, we arrive at the same conclusion noting that the
aforementioned results are obtained by adopting the group properties and not the
fact that $\Phi$ is a quantum channel. Hence this approach can be applied to
analyze the accessibility of classical bistochastic matrices by classical
dynamical semigroups generated by a Kolmogorov generator.\\

Based on the results of these work, one can easily show that a classical
dynamical semigroup $\e^{t\c K}$ is a convex combination of the group members
$T_i$. If we take the cyclic permutations as the assumed group, since cyclic
permutations are orthogonal and thus linearly independent, the classical
bistochastic matrices gained by their convex combination enjoy the following
property. The set $\mathcal{A}_N^C$ of accessible circulant bistochastic
matrices is formed by the matrices of  rank  equal to the order of possible
cyclic subgroups of $G_T$ or their multiplications. \\


The most general finite group in the classical domain is the set of all
permutations forming the Birkhoff polytope ${\cal B}_N$. The first nontrivial
case, which forms a non-abelian group, is the permutation group $S_3$. The $6$
permutation matrices and their convex combinations make the Birkhoff polytope
$\mathcal{B}_3$. This is a $4$-dimensional set comprising the convex combination
of two equilateral triangles in two orthogonal $2D$ planes. To demonstrate an
exemplary application of our approach, we discuss the set of accessible
bistochastic matrices of order three, initially analyzed in \cite{SZ18}. The
accessible set is found by exponentiating different Lindblad generators formed
by subtracting the identity from a map picked from the Birkhoff polytope ${\cal
B}_3$. The set of accessible classical channels remains in ${\cal B}_3$. We
demonstrate this set in Fig.~\ref{projections} by producing a million such
accessible channels and projecting them on the orthogonal equilateral triangles.

\begin{center}
	\begin{figure}[h]
		\includegraphics[scale = 0.15]{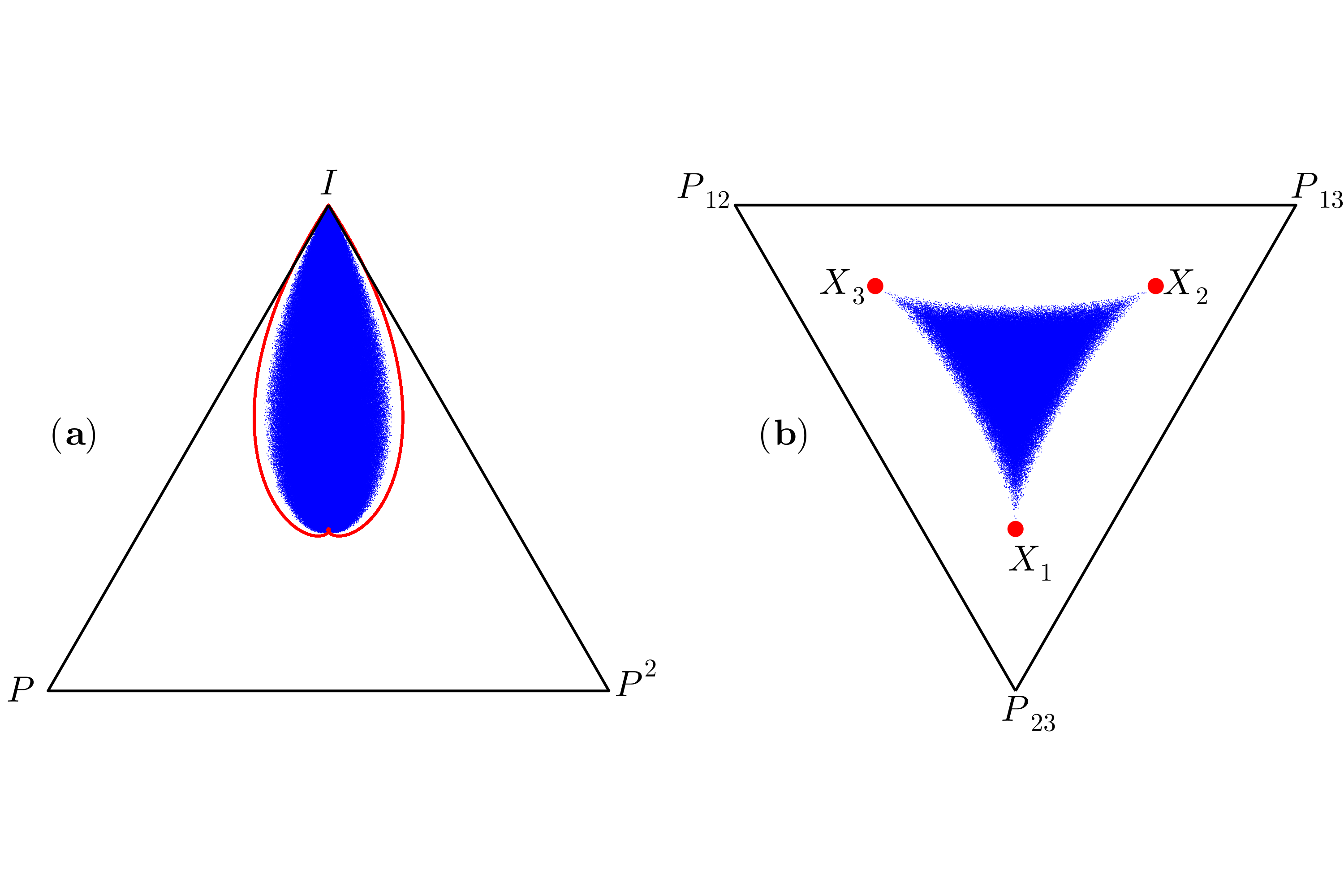}    
		\caption{The set of accessible maps in the Birkhoff polytope
			${\cal B}_3$ projected on the cross-sections corresponding to 
			(a) the even subspace, and (b) the odd subspace. $P_{ij}$ corresponds 
			to the binary permutation in $i$ and $j$ while $P$ denotes the cyclic 
			shift permutation $(312)$. The red curve in (a) represents the boundary 
			of the set $\mathcal{A}$ for the case of the cyclic group of order three
			analyzed in Sec. \ref{Examples} and displayed in Fig. \ref{two-dim} (b).			
			 The three corners of the projection on the 
			right (marked by $X_i$) are halfway between the centre of the 
			triangle and the $P_{ij}$ permutations. These are accessible by the 
			linear paths that start from the identity and end in $(\1 + P_{ij})/2$.
		}\label{projections}
	\end{figure}
\end{center}
Note that the triangle with even permutations on its corners forms the convex
hull of the cyclic group of order $3$. Therefore we expect the projection of the
accessible maps inside the $\mathcal{B}_3$ to at least cover the interior of the
red curve -- see Fig.~\ref{projections} (a).  
The data suggest that there are
no accessible maps whose projections fall \emph{outisde} this curve. The case
for the odd subspace is different as there are no accessible channels that lie
on the odd subspace other than the centre of the polytope. This is true because
the spectrum of a matrix $B = \alpha P_{23} + \beta P_{13} + \gamma P_{12}$
comprises of the values
 $\pm|\alpha + \beta e^{2\pi i/3} + \gamma e^{-2\pi
	i/3}|$ and unity. Since a positive determinant is necessary for accessibility,
we find that any combination except for $\alpha = \beta = \gamma = 1/3$ is not
accessible.\\

Another way to visualize the accessible matrices is to draw cross-sections with
2D planes parallel to the even or odd subspace with different displacement
vectors in the orthogonal subspace. Fig.~\ref{cross-sections} includes some of
these cross-sections.
\begin{center}
	\begin{figure}[h]
		\includegraphics[scale = 0.1]{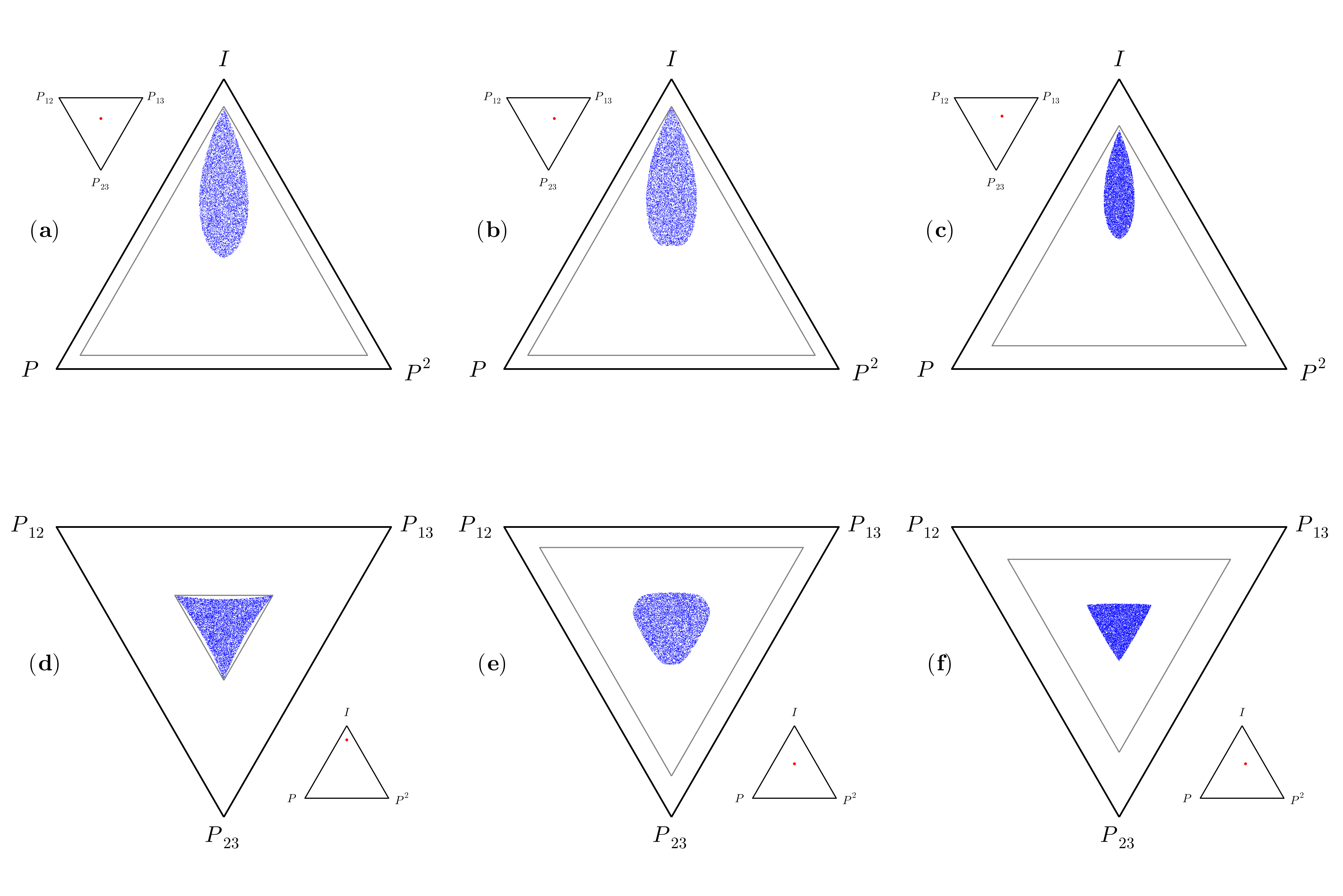}
		\caption{Several cross sections of the set of accessible points (blue) with 
		different 2D planes. a, b, c) cross-sections parallel to the even subspace; 
		d, e, f) cross-sections parallel to the odd subspace. The grey triangles 
		correspond to the cross-sections of the polytope $\mathcal{B}_3$. The red 
		dot inside the small triangles, shows the position of the planes in the 
		orthogonal subspace.}
		\label{cross-sections}
	\end{figure}
\end{center}

The volume of the Birkhoff polytope for $N=3$ is known to be 9/8~\cite{BEKTZ05}.
Using a Monte Carlo sampling method based on a sample consisting of $2.2\times10^7$ 
points that were uniformly distributed in the Birkhoff Polytope $\mathcal{B}_3$,
we estimated the relative volume of the set  ${\cal A}_3$ of accessible maps as
\begin{equation}
	\frac{\operatorname{vol}(\mathcal{A}_{3})}{\operatorname{vol}(\mathcal{B}_{3})}
	 \approx (4.398 \pm 0.004)\times10^{-2}
\end{equation}
In order to decide if a particular bistochastic map was accessible or not, we 
computed the matrix logarithm and checked if it was possible to write it as a 
positive combination of the five generators in the form $(P - \1)$.\\

\section{Concluding Remarks} \label{conclusion}

The problem of characterization of Markovian quantum channels has attracted 
considerable attention in recent years~\cite{WECC08,SCh17,Si21}.
The simplest case of single-qubit maps is already well-understood 
\cite{DZP18,PRZ19,CC19,JSP20}, while the general problem, in quantum and classical 
setup, remains open.\\



%

Following \cite{SAPZ20} we analyzed in this work the set ${\cal A}$ of
accessible quantum channels  obtained by a Lindblad generator of the form,
$L(\rho)=\sum q_i \c E_i(\rho)-\rho$, for some selected maps $\c E_i$ and
probabilities $q_i$. This set forms a subset of the set of Markovian channels of
a positive volume. Instead of analyzing the structure of  ${\cal A}$ for a fixed
dimension $N$  of the system and a concrete choice of the maps  $\c E_i$, we
studied the case in which the maps belong to a particular group $G$.
\\

The key result in this work consists in identifying properties of the set ${\cal 
A}$ of accessible maps which do not depend on the dimensionality $N$ of the
system but are determined by the properties of the group $G$. We demonstrated
that this set has a positive volume and the star-shape property with respect to
the uniform mixture of all the maps  $\c E_i$ forming the group. We computed the
volume of the set of accessible maps and analyzed its structure for several
groups -- see Fig. \ref{two-dim} showing set  ${\cal A}$ for linearly dependent
channels forming cyclic groups of order $2-7$ and Fig. \ref{4-cyc}(b) obtained for
cyclic groups of order $4$ of linearly independent channels. \\

It is worth mentioning that some of the results presented here are obtained only
based on a group's closure property. Thus, they also hold once the set $S$ of
extremal channels is a semigroup. In this way, we can also generalize the
results to the nonunital channels. For example consider the set 
$S=\{\Phi_0,\Phi_1,\Phi_2,\Phi_3,\Phi_4,\Phi_5\}$ of qubit channels where
$\Phi_0$ is the identity map, $\Phi_i$ for $i\in\{1,2,3\}$ are Pauli channels,
the nonunital maps $\Phi_4$ and $\Phi_5$ are completely contractive channels
into the states $\project0$ and $\project1$, respectively. This set is a
semigroup formed by six extreme channels. Hence, their convex hull has six
vertices. However, the elements are not all linearly independent. They satisfy
$\sum_{i=0}^{3}\Phi_i=2(\Phi_3+\Phi_4)$. In this case, the set of accessible
maps is still confined in the convex hull of the semigroup. There are, however,
some results that need other group properties. For example, note that the
trajectory $\exp[t(\sum q_\mu\Phi_\mu-\1)]$ does not necessarily end in the
centre of the polytope. A counterexample happens when all probabilities $q_\mu$
are equal to zero but $q_5$. In this case, the trajectory will end in the vertex
assigned to $\Phi_5$.
\\

Furthermore, we showed that when the vertices of the convex hull of the channels
are linearly independent, the set of accessible channels contains quantum
channels with (convexity) ranks satisfying certain constraints. The figures
mentioned above illustrate the following three facts:

(i) In the simplest case of $g=2$,  accessible maps are of ranks one and two.

(ii) In the case of $g=3$, the group of order $3$, has only two trivial
subgroups; identity and the group itself. Hence there exists an accessible map
of rank $1$, the identity map, and of rank $3$. There is no accessible map of
rank two, which correspond to the edges of the triangle.

(iii) For $g=4$, the group has  (at least) a subgroup of order two. Thus, when
we have a tetrahedron, the set  consists of maps of (convexity) rank $1$ (the
identity map), rank $2$ on the edges  assigned to the maps generating the
subgroup and maps of full rank $r=4$. Note that there are no maps of rank three
corresponding to the faces of the tetrahedrons shown in Fig. \ref{4-cyc} and
Fig.~\ref{4-ncyc}.

\medskip
We have also established more general results concerning
ranks of accessible maps acting on larger systems.

\medskip

(iv) For quantum  Weyl maps acting on $N$-level systems, the set $\cal A$
contains maps of rank $1$, $N$ and $N^2$. If $N$ is prime, this list is
complete.

(v) For a group  of Pauli channels acting on a $k$-qubit system the set $\cal A$
contains maps of rank $2^m$ with $m=0,\dots,2k$.

\medskip

Our approach can also be applied to the classical case.
We also discussed when classical and quantum accessible maps could be compared 
directly.
The present study leaves several questions open.
Let us list here some of them.

\begin{enumerate}

\item Characterize the set of points such that the set $\cal A$
       has the star-shape property with respect to them.
       





\item Find a criterion allowing one to decide whether a given channel is accessible
     with respect to a given set $S$ of quantum maps $\Phi_i$.

\item   What is the relative measure of the set  $\cal A$ of accessible channels
      with respect to Markovian channels?
%

 
%
%
\end{enumerate}


\acknowledgments{It is a pleasure to thank Seyed Javad Akhtarshenas,
Dariusz Chru{\'s}ci{\'n}ski and Kamil Korzekwa
for inspiring discussions and helpful remarks
and  David Amaro Alcal{\'a} for useful correspondence.
Financial support by Narodowe Centrum Nauki 
under the Maestro grant number DEC-2015/18/A/ST2/00274
and by Foundation for Polish Science 
under the Team-Net project no. POIR.04.04.00-00-17C1/18-00
 are gratefully acknowledged.}

\appendix
\section{The Group Matrix/Table; Proof of Proposition \ref{volumes}}
\label{AppA}

In this Appendix, we try to calculate the accessibility volume fraction for some 
abelian groups thereby providing proof for the results presented in 
Proposition \ref{volumes}.\\

The volume measure is induced from the Euclidean metric on the dynamical maps
(see \eqref{affine}). For unitary channels, or a mixture of unitary channels,
the vector $\bm t$ is zero, and we need only consider the affine matrix $M$. In
fact, the affine matrices themselves form a representation of our group of
interest. 
Furthermore, for abelian groups, we may work in a basis in which all the affine
matrices are diagonal, thereby making all the convex mixtures diagonal as well.
Finally, since we are only dealing with diagonal matrices, it sounds reasonable
to write down the diagonal entries of each affine matrix in the group in columns
of a table.
The table will therefore form a matrix with $N^2-1$ rows (remember that the
affine matrices are of order $N^2-1$ where $N$ is the dimension of the Hilbert
space.) and $|G|=g$ columns. For example, a possible group table for a qubit
representation of the non-cyclic group with four members ($\sf Z_2\times\sf
Z_2$) is the following $3\times 4$ matrix

\begin{equation}
\begin{pmatrix} +1 & +1 & -1 & -1\\
+1 & -1 & -1 & +1\\
+1 & -1 & +1 & -1
\end{pmatrix}.
\end{equation}
In general the $\sf Z_2\times \sf Z_2$ table looks like below
\begin{equation}
\begin{pmatrix} 
+\textbf{1}_a & -\textbf{1}_a & +\textbf{1}_a & -\textbf{1}_a\\
+\textbf{1}_b & -\textbf{1}_b & -\textbf{1}_b & +\textbf{1}_b\\
+\textbf{1}_c & +\textbf{1}_c & -\textbf{1}_c & -\textbf{1}_c\\
+\textbf{1}_d & +\textbf{1}_d & +\textbf{1}_d & +\textbf{1}_d
\end{pmatrix},
\end{equation}
where $\textbf{1}_a$ denotes an $a$ dimensional column vector with all its
components being equal to unity and the height of the table is given by $a+b+c+d
= N^2-1$. Note that changing the values of $a,\ b,\ c$ (and therefore $d$)
changes the shape and size of the tetrahedron; however, as we will see, the
accessibility volume fraction will remain invariant as long as the $a,b,c,d$
parameters all remain positive. Part (c) of Proposition \ref{volumes} can now be
restated as regardless of the values $(a, b, c, d)$ for the group $\sf
Z_2\times\sf Z_2$, the accessible channels cover a fraction $3/32$ of the
corresponding tetrahedron's volume.

\begin{proof}
Let $\1$ be the origin of the coordinate system and define the vectors
\begin{equation}
\textbf{e}_1\equiv(-2)\begin{pmatrix}
\textbf{1}_a\\
\textbf{1}_b\\
\textbf{0}_c\\
\textbf{0}_d
\end{pmatrix};\hspace{3mm}
\textbf{e}_2\equiv(-2)\begin{pmatrix}
\textbf{0}_a\\
\textbf{1}_b\\
\textbf{1}_c\\
\textbf{0}_d
\end{pmatrix};\hspace{3mm}
\textbf{e}_3\equiv(-2)\begin{pmatrix}
\textbf{1}_a\\
\textbf{0}_b\\
\textbf{1}_c\\
\textbf{0}_d
\end{pmatrix}.
\end{equation}
Then the tetrahedron consists of the diagonal matrices 
$\1 +\operatorname{diag}\big\{ x^1\textbf{e}_1 + x^2\textbf{e}_2 + x^3\textbf{e}_3\big\}$ 
with the constraints
\begin{equation}
x^1,x^2,x^3\geq0;\hspace{3mm}x^1+x^2+x^3\leq1 . 
\end{equation}
The line-element is $ds^2 = h_{ij}dx^i dx^j$ with the metric 
\begin{equation}
h_{ij}= \textbf{e}_i.\textbf{e}_j = 4\begin{pmatrix}
a+b & b & a\\
b & b+c & c\\
a & c & a+c
\end{pmatrix} .
\end{equation}
Therefore the volume of the tetrahedron is given by
\begin{equation}
V_{tot.} = \int \textrm{d}x^1\textrm{d}x^2\textrm{d}x^3\, \sqrt{|h|} = 
\frac{8}{3}\sqrt{abc} .
\end{equation}
It remains to calculate the volume of the accessible channels. These are found by 
element-wise exponentiation of the diagonal vectors 
$\exp\big(x^1\textbf{e}_1 + x^2\textbf{e}_2 + x^3\textbf{e}_3\big)$ 
This time the only constraint is $x^i\geq0$ and the metric is
\begin{equation}
h = 4\begin{pmatrix}
a^\prime+b^\prime & b^\prime & a^\prime\\
b^\prime & b^\prime+c^\prime & c^\prime\\
a^\prime & c^\prime & a^\prime+c^\prime
\end{pmatrix},
\end{equation}
with
\begin{equation}
a^\prime = a e^{-4(x_1+x_3)},\hspace{3mm}b^\prime = b 
e^{-4(x_1+x_2)},\hspace{3mm}c^\prime = c e^{-4(x_2+x_3)},
\end{equation}
hence
\begin{equation}
|h|=4^4(abc)e^{-8(x^1+x^2+x^3)}.
\end{equation}
Therefore, the volume is
\begin{equation}
V_{acc.} = 
\int\textrm{d}x^1\textrm{d}x^2\textrm{d}x^3\,\sqrt{|h|}=\frac14\sqrt{abc}.
\end{equation}
Finally
\begin{equation}
    \frac{V_{acc.}}{V_{tot.}} = \frac{3}{32}.
\end{equation}
\end{proof}

It is possible to repeat the same procedure for the $\sf Z_2^n$ groups to prove the 
following.\\

\begin{proposition}
	The accessibility volume ratios for $\sf Z_2^n$ are given by
	$$\frac{(2^n)!}{2^{n2^n}}$$
\end{proposition}

For 3D cyclic groups of order 4, the result is slightly different. Indeed, part
(b) of Proposition \ref{volumes} is now rephrased as for all values of $(a, b,
c)$ with $b>0$, the accessible channels cover a fraction
$\frac{3}{32}(1-e^{-4\pi})$ of the tetrahedron's volume.

\begin{proof}
    Let us start again with the group table.
    \begin{equation}
        \begin{pmatrix}\textbf{1}_a & i\textbf{1}_a & -\textbf{1}_a & -i\textbf{1}_a\\
                       \textbf{1}_a & -i\textbf{1}_a & -\textbf{1}_a & +i\textbf{1}_a\\
                       \textbf{1}_b & -\textbf{1}_b & +\textbf{1}_b & -\textbf{1}_b\\
                       \textbf{1}_c & +\textbf{1}_c & +\textbf{1}_c & +\textbf{1}_c
        \end{pmatrix}.
    \end{equation}    
    Since we are working with a real vector space, it is convenient to merge the diagonal imaginary matrix 
    \begin{equation}
        \begin{pmatrix}i & \\ & -i\end{pmatrix}
    \end{equation}
    with the real matrix
    \begin{equation}
        J\equiv \begin{pmatrix} & 1\\ -1 & \end{pmatrix} .
    \end{equation}
    The group table then becomes
    \begin{equation}
        \begin{pmatrix}\textbf{1}_{2a} & \textbf{J}_{a} & -\textbf{1}_{2a} & -\textbf{J}_a\\
    \textbf{1}_b & -\textbf{1}_b & +\textbf{1}_b & -\textbf{1}_b\\
    \textbf{1}_c & +\textbf{1}_c & +\textbf{1}_c & +\textbf{1}_c
    \end{pmatrix}
    \end{equation}
    with $2a+b+c=N^2-1$.\\
    
    Defining $\textbf{e}_i$ as we did in the previous proof, we get the metric
    \begin{equation}
        h_{ij}=\textbf{e}_i.\textbf{e}_j=4\begin{pmatrix}a+b & a & b\\
    a & 2a & a\\
    b & a & a+b\end{pmatrix},
    \end{equation}
    which leads to 
    \begin{equation}
        V_{tot.}=\frac83a\sqrt{b}.
    \end{equation}

    Now for the accessible channels, if we use the $x^i$ coordinates, we get
    $$
    \exp\big(x^1\textbf{e}_1 + x^2\textbf{e}_2 + 
    x^3\textbf{e}_3\big)=\begin{pmatrix}e^{-x^1-2x^2-x^3}\big[\cos(x^1-x^3)+J\sin(x^1-x^3)\big]_{2a}\\
    e^{-2(x^1+x^3)}\textbf{1}_b\end{pmatrix}
	$$
    However, these coordinates are not injective: the periodic functions in the 
    first part will re-address the same point many times, leading to an 
    overestimating of the volume by evaluating the integral as before. So let us 
    use $(y, z, \phi)$ defined as
    \begin{equation}
        x^1\equiv\frac{z+\phi}{2};\hspace{3mm}x^2=y;\hspace{3mm}x^3=\frac{z-\phi}{2}.
    \end{equation}
    The constraints $x^i\geq0$ become
    \begin{equation}
        y\geq0;\hspace{3mm}z\geq0;\hspace{3mm}|\phi|\leq\min(\pi, z).
    \end{equation}
    Note that the angle $\phi$ is manually constrained to stay in the interval $[-\pi, \pi)$. The metric is
    \begin{equation}
        ds^2 = 4be^{-4z}dz^2 + 2ae^{-2z-4y}\big(d\phi^2+dz^2+4dy^2+4dydz\big).
    \end{equation}
    Therefore 
    \begin{equation}
        \sqrt{|h|}=8a\sqrt{b}e^{-4z-4y}.
    \end{equation}
    Finally, the accessible volume is
    \begin{equation}
        \begin{split}
            V_{acc.}=8a\sqrt{b}\int_0^\infty e^{-4y}dy\int_0^\infty dz\, e^{-4z}\big[2\min(\pi, z)\big]\\
            =4a\sqrt{b}\Big\{\int_0^\pi ze^{-4z}dz \,+\, \pi\int_\pi^\infty e^{-4z}dz \Big\}\\
            =\frac14a\sqrt{b}\big(1-e^{-4\pi}\big).
        \end{split}
    \end{equation}
\end{proof}

A unitary channel for $N$-dimensional quantum systems has an affine map $(R, 
\textbf{0})$ where $R$ is an $N^2-1$ dimensional rotation matrix. Dealing with 
Abelian subgroups allows us to diagonalize the rotation matrices simultaneously and 
only worry about the spectral configuration of the group members, see also section 
5 of \cite{SSMZK21}.

\begin{proposition}
For the 2 dimensional cyclic group $\big\{\1,\, R,\, \cdots,\, R^{g-1}\big\}$, 
where all the members lie on a 2D plane, a fraction
\begin{equation}
\frac{1-e^{-2\pi\tan(\pi/g)}}{2g\sin^2(\pi/g)}= 1 - \frac{\pi^2}{g}+\mathcal{O}(1/g^2)
\end{equation}
of the 2D regular polygon is covered by exponentiating the Lindblad generators.
\end{proposition}
\begin{proof}
Forgetting about the scale of the actual polygon, we may identify it with the complex polygon with vertices $\{1, e^{2\pi i/g}, \cdots, e^{-2\pi i/g }\}$. Then the area of the polygon is
\begin{equation}
A = g\sin(\pi/g)\cos(\pi/g).
\end{equation}
By exponentiating the Lindblad generators, we get the complex set
\begin{equation}
\big\{e^{-x+iy}\,\big|\,x\geq 0,\,|y|<\frac{x}{\tan(\pi/n)}\big\} = 
\big\{re^{i\phi}\,\big|\,0\leq r\leq 1,\,|\phi|<\frac{-\log r}{\tan(\pi/n)}\big\}.
\end{equation}
A standard calculation of the area of this set leads to the claimed result.\\
Observation: Any complex number $z$ with $|z|<r^* = e^{-\pi\tan(\pi/g)}$ is infinitesimally divisible. It may be interesting to calculate $r^*$ for other groups as well.
\end{proof}

\medskip

\end{document}